\documentclass[12 pt, a4paper]{article}

\usepackage{graphicx} 
\usepackage[utf8]{inputenc}
\usepackage[T1]{fontenc}
\usepackage{lmodern}
\usepackage{tgtermes}
\usepackage{geometry}
\usepackage{amsmath} 
\allowdisplaybreaks
\usepackage{amsthm}
\usepackage{amssymb}
\usepackage{dsfont}
\usepackage{float}
\usepackage{color}
\usepackage[english]{babel}
\usepackage[ruled,vlined]{algorithm2e}
\usepackage{multirow}
\usepackage{csquotes}
\usepackage{setspace}
\usepackage{longtable}
\usepackage{caption}
\usepackage{subcaption}
\usepackage{tikz}
\usepackage{emptypage}
\usepackage{bm}
\usepackage{hhline}
\usepackage{diagbox}
\usepackage[colorlinks, citecolor=blue]{hyperref}
\usepackage[authoryear, round]{natbib}
\usepackage{lineno}

\newtheorem{theorem}{Theorem}

\newtheorem{corollary}[theorem]{Corollary}

\newtheorem{definition}{Definition}

\title{Identifiability of linear stochastic state-space models with application to ecology}

\author{Fr\'ed\'eric Barraquand and Julien Gibaud$^*$
 \\ \\
Institute of Mathematics of Bordeaux, University of Bordeaux,\\ CNRS, Bordeaux INP, Talence, France\\}

\date{}

\begin{document}

\maketitle


\begin{abstract}
 State-space models are dynamical systems defined by a latent and an observed process. In ecology, stochastic state-space models in discrete time are most often used to describe the imperfectly observed dynamics of population sizes or animal movement. However, several studies have observed identifiability issues when state-space models are fitted to simulated or real data, and it is not currently clear whether those are due to data limitations or more fundamental model non-identifiability. To investigate such theoretical identifiability, a suitable exhaustive summary is required, defined as a vector of parameter combinations which fully determines the model. Previous work on exhaustive summaries has used expectations of the stochastic process, so that noise parameters are unaccounted for. In this paper, we build an exhaustive summary using the spectral density of the observed process, which fully accounts for all mean and variance parameters. This diagnostic is applied to contrasted ecological models and we show that they are generally theoretically identifiable, unless some model compartements are unobserved. This suggest that issues encountered while fitting models are mostly due to practical identifiability.
\end{abstract}



\noindent  {\small  $^*$ Correspondence to julien.gibaud@math.u-bordeaux.fr}

\thispagestyle{empty}
 
\newpage

\section{Introduction}

State-Space Models (SSMs) are deterministic or stochastic dynamical systems defined by two processes. The state process, which is not observed directly, models the transformation of the states over time. By contrast, the observation process describes the relationship between current states and the observables on which model fitting and prediction are based \citep{aoki2013state}. 
Ecology and related disciplines such as fishery sciences, conservation and other environmental fields frequently use stochastic SSMs to represent the imperfectly observed dynamics of population sizes \citep{Ives2003estimating} or animal movement \citep{jonsen2005robust}. In a systematic review, \citet{AugerMethe2021Aguide} explore the different aspects of state-space modeling in ecology, and remark the frequent identifiability issues encountered by practitioners. Indeed, several simulation-based studies, such as \citet{valpine2005State}, \citet{knape2008estimability} or \citet{AugerMethe2016state} have demonstrated practical identifiability issues in ecological SSMs.  A main objective of the present paper is to determine which of the observed issues are due to theoretical non-identifiability. 

Tackling (non)identifiability issues requires first a definition of identifiability. Let us consider that a model with $p$ parameters $\bm{\theta}=(\theta_1,\dots,\theta_p)$ can be represented by a suitable function of parameters $\bm{\theta}\mapsto M(\bm{\theta})$. For instance, for a parametric statistical model, this function can be the (log-)likelihood. A model is then said to be globally identifiable if the function $M$ is injective for all the parameters, i.e $M(\bm{\theta}_1)=M(\bm{\theta}_2)$ implies that $\bm{\theta}_1=\bm{\theta}_2$ for all $\bm{\theta}_1$ and $\bm{\theta}_2$ in the parameter space \citep{cole2010determining}. If there exists an open neighbourhood $\Omega$ of any $\bm{\theta}$ such that the function $M$ is injective, the model is locally identifiable.
In a linear-time-invariant SSM context, identifiability can be determined through several methods. Assuming that the initial condition is at rest ($\mathbf{x}_0 = \bf{0}$), the dynamical systems whose parameters are uniquely given by the impulse response or the transfer function are identifiable \citep{bellman1970on}. Additionally, \citet{grewal1976identifiability} and \citet{thowsen1978identifiability} show that the model is identifiable if the rank of the Jacobian of the mapping from the parameter space into the so-called Markov parameters is equal to the number of parameters. This is an alternative route to proving identifiability than using transfer functions. Yet another way of assessing SSM identifiability consists in using the infinite set of equations equating the successive derivatives of the observation with a constant. If the system has an unique solution, we have a sufficient condition for identifiability. This approach is called the Taylor series expansion \citep{pohjanpalo1978system}. An early review of these methods is provided by \citet{godfrey1985identifiability}, with some examples. 
Ecology and related disciplines producing parameter-rich models for complex systems have been particularly interested in knowing whether all or part of SSM parameters can be estimated \citep[][Chapter 9]{cole2020parameter}. To that aim, \citet{cole2016parameter} propose to investigate the identifiability of SSMs by means of exhaustive summaries, which are simplified representations of the model that can be obtained from the multiple techniques listed above.  

An exhaustive summary $\bm{\kappa}\in \mathbb{R}^n$ is a vector of parameter combinations which fully determines the model, i.e. $\bm{\kappa}(\bm{\theta}_1)=\bm{\kappa}(\bm{\theta}_2)$ if and only if $M(\bm{\theta}_1)=M(\bm{\theta}_2)$ for all $\bm{\theta}_1$ and $\bm{\theta}_2$ in the parameter space. From these definitions, \citet[][Theorem 1]{cole2010determining}  show that a model is identifiable (locally or globally) if the exhaustive summary is injective for the parameters. The non-identifiability of the model occurs when the model is parameter redundant \citep{catchpole1997detecting}. A model is parameter redundant if $\bm{\kappa}(\bm{\theta})$ can be expressed in terms of a smaller parameter vector $\bm{\beta}=(\beta_1,\dots,\beta_q)$, where $q<p$ and $\bm{\beta}=f(\bm{\theta})$ for a given function $f$. \citet[][Theorem 2a]{cole2010determining} show that a model is parameter redundant if the rank of the derivative matrix $\partial \bm{\kappa}/\partial \bm{\theta}$ is lower than the number of parameters. For deterministic continuous-time SSMs, an exhaustive summary based on either the transfer function or the Taylor series expansion can be built \citep[][Section 7.2]{cole2020parameter}. However, SSMs come in various flavours and ecologists (as well as many applied statisticians) are mostly interested in stochastic and discrete-time versions of SSMs. 
However, building a good exhaustive summary for a stochastic SSM is not completely straightforward based on the current literature. The strategy adopted by \citet{cole2016parameter} was to create exhaustive summaries based on the expectations of the stochastic process. This expansion technique consists in calculating the successive expectations of the observables at each time-step. To supplement the exhaustive summary by variance terms, \citet[][Appendix A.4]{cole2016parameter} proposed to expand the variance as well as the expectation.
In population dynamics, recent works have applied this method to demographic models \citep{polansky2020improving, aldrin2021caveats, polansky2023combining}.
The two approaches, working solely with expectations of the stochastic processes, as well as adding variance of observations to construct exhaustive summaries, have limitations when there are process and observation error parameters.  The first approach does not allow to take into account the possible dependency between variance parameters and (deterministic) process parameters, since it neglects variance altogether. Whilst as observed by \citet{knape2008estimability} on simulated data, the dependency between variance and non-variance parameters influences the (un)identifiability of SSMs, even simple linear ones. The second approach, adding variances to expectations of the observations, seems sensible but has not been justified mathematically---in the discussion and Supplementary Material, we show that the approach is partially justified but not exact in general, as some terms can be missing from the exhaustive summaries thus constructed. To rigourously evaluate the identifiability of ecological SSMs, a fully stochastic approach is therefore needed. 

Inspired by the frequent use of Fourier techniques in the analysis of stochastic processes, we sought to move from the deterministic transfer function approach to identifiability assessment to a stochastic equivalent. We show that the output spectral density plays a key role in stochastic SSM identifiability assessment, building upon early work by \citet{glover1974parametrizations}. This allows us to define new exhaustive summaries for fully stochastic models, without resorting to a deterministic approximation as previously done in the ecological statistics literature. 
The new exhaustive summary show that usual ecological models are often theoretically identifiable, suggesting that frequent SSM estimation problems are due to practical rather than theoretical identifiability issues, though we also show that several reasonable variants of common models can be theoretically unidentifiable.

The paper is organized as follows. In Section \ref{Definition}, we recall how stochastic SSMs are defined. Section \ref{MainTheorem} presents the new exhaustive summary built on the output spectral density. Section \ref{Application} then goes through contrasted applications to ecological models that illustrate the interest of the proposed exhaustive summary. Finally, Section \ref{Discussion} sums up what we learned about variance parameter identifiability issues and improvements provided by the new exhaustive summary, together with a discussion of some limitations of the present study.

\section{Definition of state-space models}
\label{Definition}

State-space models (SSMs) are dynamical systems that are continuous or discrete in time, and can be deterministic or stochastic. In this paper, we consider the linear-time-invariant (LTI) discrete-time Gaussian stochastic SSM defined by two processes. The state equation
\begin{equation}
\label{StateEquation}
    \bm{x}_{t+1}=\bm{Ax}_t+\bm{Bu}_t
\end{equation}
gives the sequence $(\bm{x}_t)$ of the true $N$ latent states while the observation equation
\begin{equation}
\label{ObservationEquation}
    \bm{y}_t=\bm{Cx}_t+\bm{Du}_t
\end{equation}
leads to the output sequence $(\bm{y}_t)$ of the observations. The input sequence $(\bm{u}_t)$ is assumed to be a white noise process, such that for all $t \geq 0$, $\bm{u}_t\sim \mathcal{N}(\bm{0}, \bm{Q})$. The vector $\bm{u}_t$ encapsulates both measurement and process errors (the dimension of $\bm{u}_t$ can be larger than that of $\bm{x}_t$), so that independence between observation and process errors is deal by choosing carefully matrices $\bm{B}$ and $\bm{D}$ (see Section \ref{Application}). The matrices $\bm{A}\in \mathbb{R}^{N\times N}$, $\bm{B}\in \mathbb{R}^{N\times P}$, $\bm{C}\in \mathbb{R}^{M\times N}$, $\bm{D}\in \mathbb{R}^{M\times P}$ and $\bm{Q}\in \mathbb{R}^{P\times P}$ contain all the parameters $\bm{\theta}$ to estimate in the model.
As assumed by \citet{glover1974parametrizations}, in this paper the output $(\bm{y}_t)$ is a stationary process (i.e. all the eigenvalues of $\bm{A}$ strictly belong to unit sphere) and the system has reached steady state when the observations begin.
Originally used for continuous time, the transfer function can also used to represent discrete-time dynamic systems \citep{bellman1970on}. Based on the discrete-time Fourier transform of $(\bm{y}_t)$ defined as
\begin{equation}
   \bm{Y}(i\omega)=\mathcal{F}\{(\bm{y}_t)\}(i\omega)=\sum_{t=-\infty}^{+\infty}\bm{y}_t\exp(-i\omega t),\quad i\omega\in i\mathbb{R}, 
\end{equation}
and assuming that the initial condition is at rest (i.e. $\bm{x}_0=\bm{0}$), the transfer function writes
\begin{equation}
\label{TransferFunction}
    \bm{H}_s(\bm{\theta})=\bm{C}(s\bm{I}-\bm{A})^{-1}\bm{B}+\bm{D}, \quad \text{with} \quad s=\exp(i\omega).
\end{equation}
In absence of variance parameters, \citet{cole2016parameter} use this expression to build a suitable exhaustive summary. However, if variance parameters are involved in the model, a function of all the parameters need to be exploited. 
This function is provided by the output spectral density, which is then given by
\begin{equation}
    \bm{S}(\bm{\theta})=\mathbb{E}\left[\bm{Y}(i\omega)\bm{Y}(-i\omega)^{\top}\right]=\bm{H}_{s}(\bm{\theta})\bm{Q}\bm{H}_{1/s}^{\top}(\bm{\theta}).
\end{equation}
Inspired by \citet[][Section 2.4 and Section 6.6]{sarkka2019applied}, we give the derivations of the transfer function and the output spectral density in Supplementary Material \ref{derivations}.

\section{The output spectral density exhaustive summary}
\label{MainTheorem}

The output spectral density gives a complete description of the stochastic process \citep{glover1974parametrizations}. We can therefore adopt the following definition
\begin{definition}
    The stochastic linear-time-invariant discrete-time  SSM is said to be locally identifiable if 
    \begin{equation}
        \bm{S}(\bm{\theta}_1)=\bm{S}(\bm{\theta}_2)\implies \bm{\theta}_1=\bm{\theta}_2, \quad \forall \bm{\theta}_1,\bm{\theta}_2\in \Omega,
    \end{equation}
    where $\Omega\subset \mathbb{R}^p$ is an open set.
\end{definition}
Using this definition, we have the main result of this paper.
\begin{theorem}
\label{OSDES}
    For stochastic linear-time-invariant discrete-time SSMs, the components of the output spectral density have the form
    \begin{equation}
        [\bm{S}(\bm{\theta})]_{ij}=\dfrac{\sum_{n=0}^{2N}a_n(\bm{\theta})s^{n}}{\sum_{n=0}^{2N}b_n(\bm{\theta})s^{n}},
    \end{equation}
    and the output spectral density exhaustive summary $\bm{\kappa}(\bm{\theta})$ consists of the non-constant coefficients $a_n(\bm{\theta})$ and $b_n(\bm{\theta})$ of the powers of $s$ in the numerators and denominators of $\bm{S}(\bm{\theta})$.
\end{theorem}

\begin{proof}[Proof of the theorem]
    From \citet[][Section 7.2.1]{cole2020parameter}, we know that the components of the transfer function have the form
    \begin{equation}
        [\bm{H}_s(\bm{\theta})]_{ij}=\dfrac{\sum_{n=0}^N \alpha_n(\bm{\theta})s^n}{\sum_{n=0}^N \beta_n(\bm{\theta})s^n}.  
    \end{equation}
    The formulation of the transfer function (\ref{TransferFunction}) implies that the denominator is the characteristic polynomial $\chi_{A}$ of $\bm{A}$ evaluated in $s$. The transfer function could be rewritten
    \begin{equation}
        \bm{H}_s(\bm{\theta})=\dfrac{1}{\chi_{A}(s)}\bm{G}_s(\bm{\theta}),\quad \text{where} \quad [\bm{G}_s(\bm{\theta})]_{ij}= \sum_{n=0}^N \alpha_n(\bm{\theta})s^n.   
    \end{equation}
    Since $s=\exp(i\omega)$ and the modulus of the eigenvalues of $\bm{A}$ strictly lower than 1, the transfer function is always defined. The output spectral density thus writes
    \begin{equation}
        \bm{S}(\bm{\theta})=\bm{H}_s(\bm{\theta})\bm{Q}\bm{H}_{1/s}^{\top}(\bm{\theta})=\dfrac{1}{\chi_{A}(s)\chi_{A}(1/s)}\bm{G}_s(\bm{\theta})\bm{Q}\bm{G}_{1/s}^{\top}(\bm{\theta}).
    \end{equation}
    The product $\bm{G}_s(\bm{\theta})\bm{Q}\bm{G}_{1/s}^{\top}(\bm{\theta})$ leads to a matrix where each component  is a sum of a polynomial of degree $N$ in $s$ and a polynomial of degree $N$ in $1/s$ since the matrix $\bm{Q}$ does not change the degree of the polynomials. Moreover, $\chi_{A}(s)\chi_{A}(1/s)$ is also a sum of polynomials. The components of the spectral density have then the form
    \begin{equation}
        [\bm{S}(\bm{\theta})]_{ij}=\dfrac{\alpha_0(\bm{\theta})+\sum_{n=1}^N\alpha_n(\bm{\theta})s^n+\sum_{n=1}^N\gamma_n(\bm{\theta})\frac{1}{s^n}}{\beta_0(\bm{\theta})+\sum_{n=1}^N\beta_n(\bm{\theta})s^n+\sum_{n=1}^N\delta_n(\bm{\theta})\frac{1}{s^n}}.
    \end{equation}
    Multiplying $\bm{S}(\bm{\theta})$ by $s^N/s^N$, the components of the spectral density have the expected form. As explained by \citet{bellman1970on} and \citet{godfrey1985identifiability} for the transfer function, if there are no common factors in these numerator and denominator polynomials, the coefficients are uniquely given by the spectral density. The rational function being entirely identified by their coefficients, after removing potential common factors the exhaustive summary encapsulates all the coefficients of all the polynomials of the output spectral density.
\end{proof}

Analogous result could be shown for the continuous counterpart of the stochastic SSMs. Supplementary Material \ref{ContinuousCounterpart} presents this alternative.

\section{Application to ecological models}
\label{Application}
In this section, we investigate the (non-)identifiability of different ecological SSMs through the symbolic differentiation approach. This approach consists in calculating the derivative matrix $\partial \bm{\kappa} / \partial \bm{\theta}=(\partial \kappa_j / \partial \theta_i)_{ij}$ for $i=1,\dots,p$ and $j=1,\dots,n$. Now, if the rank $r$ of the derivative matrix is lower than the number of parameters $p$, the model is said to be parameter redundant \citep{catchpole1997detecting}. On the other hand, if $r=p$, the model is full rank and the parameters are theoretically estimable. When a model is parameter redundant, we define the deficiency as $d=p - r$. The equation $\bm{\alpha}^{\top}(\partial \bm{\kappa} / \partial \bm{\theta})=\bm{0}$ gives $d$ vectors $(\bm{\alpha}_j, j=1,\dots,d)$ characterizing a basis of the null space of $(\partial \bm{\kappa} / \partial \bm{\theta})^{\top}$. In order to find the $r$ estimable parameters, we need to solve the system of linear first-order partial differential equations $\sum_{i=1}^p \alpha_{ij}(\partial f/\partial \theta_i)=0$, for $j=1,\dots,d$, where $f$ is an arbitrary function. All these results are detailed by \citet[][Theorem 2]{cole2010determining} and \citet[][Section 3.2]{cole2020parameter}.

The symbolic algebra of this approach is executed through the software application Maple 2024.1. The Maple files for the application are freely available following the link 
\url{https://github.com/julien-gibaud/IdentifiabilitySSM}.

\subsection{Population dynamics}
\label{PopulationDynamics}

We illustrate Theorem \ref{OSDES}  with a model derived from \citet{Besbeas2002Integrating} and already used in \citet{cole2020parameter} with exhaustive summaries based on process expectations. They create a discrete-time SSM for log-abundance data on the number of lapwings (\textit{Vanellus vanellus}). The state variable is $\bm{x}_t = (N_{1,t}, N_{a,t})^{\top}$, where $N_{1,t}$ is the number of juvenile birds and $N_{a,t}$ is the number of adult birds in year $t$. The underlying state equation describes the birth of new birds, survival from one year to the next and the process of moving from a juvenile bird to an adult bird, as following
\begin{equation}
    \bm{x}_{t+1}=\begin{pmatrix}
    0 & \rho \phi_1\\
    \phi_a & \phi_a
\end{pmatrix}\bm{x}_{t}+\bm{\varepsilon}_t,\quad \bm{\varepsilon}_t\sim \mathcal{N}(\bm{0},\sigma^2_{\varepsilon}\bm{I}_2)
\end{equation}
where $\rho$ denotes the annual productivity rate, $\phi_1$ is the probability that a bird survives its first year and $\phi_a$ is the annual survival probability for adult birds. As only adult birds are observed the observation equation is
\begin{equation}
    y_t=\begin{pmatrix}
    0 & 1
\end{pmatrix}\bm{x}_t+\eta_t,\quad \eta_t\sim \mathcal{N}(0,\sigma^2_{\eta}).
\end{equation}
The observation variable $y_t$ then represents the observed numbers of adult birds at time $t$. The vector of parameters is $\bm{\theta}=(\rho, \phi_1,\phi_a, \sigma^2_{\eta}, \sigma^2_{\varepsilon})$. Since the parameters $\rho$ and $\phi_1$ only appear as a product, it is clear that these parameters are confounded and the model not identifiable. This case study therefore aims at showing how the exhaustive summary behave in a case we know not to be identifiable. In order to obtain the state-space model similar to Equations (\ref{StateEquation}) and (\ref{ObservationEquation}), the matrices need to be 
\begin{equation}
\bm{A}=\begin{pmatrix}
    0 & \rho\phi_1\\
    \phi_a & \phi_a 
\end{pmatrix}, \bm{B}=\begin{pmatrix}
    1 & 0 & 0\\
    0 & 1 & 0
\end{pmatrix},\bm{C}=\begin{pmatrix}
    0 & 1
\end{pmatrix},\bm{D}=\begin{pmatrix}
    0 & 0 & 1
\end{pmatrix}
\end{equation}
and the input process is given by 
\begin{equation}
    \forall t\geq 0,\quad \bm{u}_t:= \begin{pmatrix}
    \bm{\varepsilon}_t\\
    \eta_t
    \end{pmatrix}\sim \mathcal{N}(\bm{0},\bm{Q}), \quad \text{with} \quad
    \bm{Q}=\begin{pmatrix}
    \sigma^2_{\varepsilon} & 0 & 0\\
    0 & \sigma^2_{\varepsilon} & 0\\
    0 & 0 & \sigma^2_{\eta}
    \end{pmatrix}.
\end{equation}
The coefficients of the output spectral density are given by 
\begin{equation}
    \begin{pmatrix}
    a_0(\bm{\theta})\\
    a_1(\bm{\theta})\\
    a_2(\bm{\theta})\\
    a_3(\bm{\theta})\\
    a_4(\bm{\theta})\\
    b_0(\bm{\theta})\\
    b_1(\bm{\theta})\\
    b_2(\bm{\theta})\\
    b_3(\bm{\theta})\\
    b_4(\bm{\theta})\\
    \end{pmatrix}=\begin{pmatrix}
    -\rho \phi_1\phi_a\sigma^2_{\eta}\\
    \rho\phi_1\phi_a^2\sigma^2_{\eta}-\phi_a\sigma^2_{\eta}\\
    \rho^2\phi_1^2\phi_a^2\sigma^2_{\eta}+\phi_a^2\sigma^2_{\eta}+\phi_a^2\sigma^2_{\varepsilon}+\sigma^2_{\eta}+\sigma^2_{\varepsilon}\\
    \rho\phi_1\phi_a^2\sigma^2_{\eta}-\phi_a\sigma^2_{\eta}\\
    -\rho \phi_1\phi_a\sigma^2_{\eta}\\
    -\rho \phi_1\phi_a\\
    \rho \phi_1\phi_a^2-\phi_a\\
    \rho^2\phi_1^2\phi_a^2+\phi_a^2+1\\
    \rho \phi_1\phi_a^2-\phi_a\\
    -\rho \phi_1\phi_a
    \end{pmatrix}.
\end{equation}
Removing the equal coefficients, the exhaustive summary writes
\begin{equation}
    \bm{\kappa}(\bm{\theta})=\begin{pmatrix}
    -\rho \phi_1\phi_a\sigma^2_{\eta}\\
    \rho\phi_1\phi_a^2\sigma^2_{\eta}-\phi_a\sigma^2_{\eta}\\
    \rho^2\phi_1^2\phi_a^2\sigma^2_{\eta}+\phi_a^2\sigma^2_{\eta}+\phi_a^2\sigma^2_{\varepsilon}+\sigma^2_{\eta}+\sigma^2_{\varepsilon}\\
    -\rho \phi_1\phi_a\\
    \rho \phi_1\phi_a^2-\phi_a\\
    \rho^2\phi_1^2\phi_a^2+\phi_a^2+1
    \end{pmatrix}.
\end{equation}
Now, we compute the derivative matrix 
\begin{align}
    &\dfrac{\partial \bm{\kappa}}{\partial \bm{\theta}} = \\& \begin{pmatrix}
    -\phi_1\phi_a\sigma^2_{\eta} & \phi_1\phi_a^2\sigma^2_{\eta}   & d_{13} & -\phi_1\phi_a & \phi_1\phi_a^2 & 2\rho\phi_1^2\phi_a^2\\
    -\rho \phi_a\sigma^2_{\eta}  & \rho\phi_a^2\sigma^2_{\eta}  & d_{23} & -\rho\phi_a & \rho\phi_a^2 & 2\rho^2\phi_1\phi_a^2\\
    -\rho \phi_1\sigma^2_{\eta}  & 2\rho\phi_1\phi_a\sigma^2_{\eta}-\sigma^2_{\eta}   & d_{33} & -\rho\phi_1 & 2\rho\phi_1\phi_a-1 & 2\rho^2\phi_1^2\phi_a+2\phi_a\\
    -\rho \phi_1\phi_a     & \rho\phi_1\phi_a^2-\phi_a  & d_{43} & 0 & 0 & 0\\
    0     & 0    & d_{53} & 0 & 0 & 0
    \end{pmatrix}, \nonumber
\end{align}
where $d_{13}=2\rho\phi_1^2\phi_a^2\sigma^2_{\eta}$, $d_{23}=2\rho^2\phi_1\phi_a^2\sigma^2_{\eta}$, $d_{33}=2\rho^2\phi_1^2\phi_a\sigma^2_{\eta} +  2\phi_a\sigma^2_{\eta}+2\phi_a\sigma^2_{\varepsilon}$, $d_{43}=\rho^2\phi_1^2\phi_a^2 +\phi_a^2 + 1$ and $d_{53}=\phi_a^2+1$. As the derivative matrix has rank $r=4$, but there are five parameters, this model is parameter redundant with deficiency $d=1$. Solving the equation $\bm{\alpha}^{\top} \partial \bm{\kappa}/\partial \bm{\theta}=\bm{0}$ gives
\begin{equation}
    \bm{\alpha}^{\top}=\begin{pmatrix}
    -\dfrac{\rho}{\phi_1} & 1 & 0 & 0 & 0
\end{pmatrix}.
\end{equation}
The position of the zeros indicates that the  parameters $\phi_a$, $\sigma^2_{\eta}$ and $\sigma^2_{\varepsilon}$ are individually identifiable. Solving the partial differential equation
\begin{equation}
    -\dfrac{\rho}{\phi_1}\dfrac{\partial f}{\partial \phi_1}+\dfrac{\partial f}{\partial \rho}=0
\end{equation}
gives $\beta_1= \rho\phi_1$, $\beta_2=\phi_a$, $\beta_3=\sigma^2_{\eta}$ and $\beta_4=\sigma^2_{\varepsilon}$ as the estimable parameters.

\subsection{The state-space AR(1) model}
\label{AR1}

The autoregressive model of order one in state-space form is classically used in ecology to model population dynamics \citep{Ives2003estimating, dennis2006estimating, knape2008estimability}.

\subsubsection{Model without interaction between state and observation processes}

\citet{AugerMethe2016state} consider the simple AR(1) model
\begin{equation}
  x_{t+1}=\rho x_t +\varepsilon_t, \quad \varepsilon_t \sim \mathcal{N}(0,\sigma^2_{\varepsilon})  
\end{equation}
with an observation equation
\begin{equation}
    y_t=x_t+\eta_t,\quad \eta_t \sim \mathcal{N}(0,\sigma^2_{\eta}).
\end{equation}
We thus have $A=\rho$, $B=(1\,\,0)$, $C=1$, $D=(0\,\,1)$ and
\begin{equation}
    \bm{Q}=\begin{pmatrix}
    \sigma^2_{\varepsilon} & 0\\
    0 & \sigma^2_{\eta}
\end{pmatrix}.
\end{equation}
The vector of parameters is $\bm{\theta}=(\rho, \sigma^2_{\eta}, \sigma^2_{\varepsilon})$. The output spectral density is given by
\begin{equation}
    \bm{S}(\bm{\theta})=\dfrac{-\rho\sigma^2_{\eta}s^2+(\rho^2\sigma^2_{\eta}+\sigma^2_{\eta}+\sigma^2_{\varepsilon})s-\rho\sigma^2_{\eta}}{-\rho s^2+( \rho^2+1)s-\rho}.
\end{equation}
Removing the equal coefficients of the output spectral density, the exhaustive summary writes
\begin{equation}
    \bm{\kappa}(\bm{\theta})=\begin{pmatrix}
    -\rho\sigma^2_{\eta}\\
    \rho^2\sigma^2_{\eta}+\sigma^2_{\eta}+\sigma^2_{\varepsilon}\\
    -\rho\\
    \rho^2+1
\end{pmatrix}.
\end{equation}
The derivative matrix
\begin{equation}
    \dfrac{\partial\bm{\kappa}}{\partial \bm{\theta}}=\begin{pmatrix}
    -\sigma^2_{\eta} & 2\rho\sigma^2_{\eta} & -1 & 2\rho\\
    -\rho & \rho^2+1 & 0 & 0\\
    0 & 1 & 0 & 0
\end{pmatrix}
\end{equation}
has rank $r=3$, and there are three parameters, this model is not parameter redundant.

\subsubsection{Model with interaction between processes}
We consider the above model with this time a non-diagonal
\begin{equation}
    \bm{Q}=\begin{pmatrix}
    \sigma^2_{\varepsilon} & \sigma_{\varepsilon,\eta}\\
    \sigma_{\varepsilon,\eta} & \sigma^2_{\eta}
\end{pmatrix},
\end{equation}
where $\sigma_{\varepsilon,\eta}$ models the interaction between the two processes. The vector of parameters becomes $\bm{\theta}=(\rho, \sigma^2_{\eta}, \sigma^2_{\varepsilon}, \sigma_{\varepsilon,\eta})$. The output spectral density is given by
\begin{equation}
    \bm{S}(\bm{\theta})=\dfrac{(-\rho\sigma^2_{\eta}+\sigma_{\varepsilon,\eta})s^2+(\rho^2\sigma^2_{\eta}-2\sigma_{\varepsilon,\eta}\rho+\sigma^2_{\eta}+\sigma^2_{\varepsilon})s-\rho\sigma^2_{\eta}+\sigma_{\varepsilon,\eta}}{-\rho s^2+( \rho^2+1)s-\rho}.
\end{equation} 
Finally, the exhaustive summary writes
\begin{equation}
    \bm{\kappa}(\bm{\theta})=\begin{pmatrix}
    -\rho\sigma^2_{\eta}+\sigma_{\varepsilon,\eta}\\
    \rho^2\sigma^2_{\eta}-2\sigma_{\varepsilon,\eta}\rho+\sigma^2_{\eta}+\sigma^2_{\varepsilon}\\
    -\rho\\
    \rho^2+1
\end{pmatrix}.
\end{equation}
The derivative matrix 
\begin{equation}
   \dfrac{\partial\bm{\kappa}}{\partial \bm{\theta}}=\begin{pmatrix}
    -\sigma^2_{\eta} & 2\rho\sigma^2_{\eta}-2\sigma_{\varepsilon,\eta} & -1 & 2\rho\\
    -\rho & \rho^2+1 & 0 & 0\\
    0 & 1 & 0 & 0\\
    1 & -2\rho & 0 & 0
\end{pmatrix} 
\end{equation}
has rank $r=3$, but there are four parameters, this model is parameter redundant with deficiency $d=1$. Solving the equation $\bm{\alpha}^{\top}\partial \bm{\kappa}/\partial \bm{\theta}=\bm{0}$ gives
\begin{equation}
    \bm{\alpha}^{\top}=\begin{pmatrix}
    0 & \dfrac{1}{\rho} & \dfrac{\rho^2-1}{\rho} & 1
\end{pmatrix}.
\end{equation}
Solving the partial differential equation
\begin{equation}
    \dfrac{1}{\rho}\dfrac{\partial f}{\partial \sigma^2_{\eta}}+\dfrac{\rho^2-1}{\rho}\dfrac{\partial f}{\partial \sigma^2_{\varepsilon}}+\dfrac{\partial f}{\partial \sigma_{\varepsilon,\eta}}=0
\end{equation}
gives $\beta_1=\rho$, $\beta_2=-\rho^2\sigma_{\eta}^2+\sigma^2_{\eta}+\sigma^2_{\varepsilon}$ and $\beta_3=-\rho\sigma_{\eta}^2+\sigma_{\varepsilon,\eta}$ as estimable parameters.

\subsection{Animal movement}
\label{AnimalMovement}

As a third example, \citet{AugerMethe2016state} propose to model the movement of polar bears (\textit{Ursus maritimus}) by taking into account the sea ice drift. The two-dimensional observed variable $\bm{y}_t$ measures the daily displacement of the polar bear based on the GPS collar, the state variable $\bm{x}_t$ is the active movement of the polar bear while the sea ice drift is encoded by the standardized covariate $\bm{s}_t$. The state equation writes
\begin{equation}
    \bm{x}_{t+1}=\begin{pmatrix}
    \rho_u & 0\\
    0 & \rho_v
\end{pmatrix}\bm{x}_t+\bm{\varepsilon}_t,\quad \bm{\varepsilon}_t\sim \mathcal{N}(\bm{0},\bm{Q}_1) \quad \text{with} \quad \bm{Q}_1=\begin{pmatrix}
    \sigma^2_{\varepsilon, u} & 0\\
    0 & \sigma^2_{\varepsilon, v}
\end{pmatrix}
\end{equation}
where the parameters $\rho_u$ and $\rho_v$ are degree of autocorrelation in the random walk. The observation equation is given by
\begin{equation}
    \bm{y}_t=\bm{x}_t+\bm{s}_t+\bm{\eta}_t, \quad \bm{\eta}_t\sim \mathcal{N}(\bm{0},\bm{Q}_2)\quad \text{with} \quad \bm{Q}_2=\begin{pmatrix}
    \sigma^2_{\eta, u} & 0\\
    0 & \sigma^2_{\eta, v}
\end{pmatrix}.
\end{equation}
The vector of parameters is $\bm{\theta}=(\rho_u, \rho_v, \sigma^2_{\eta,u}, \sigma^2_{\eta,v}, \sigma^2_{\varepsilon,u}, \sigma^2_{\varepsilon,v})$. As above, we rewrite this model through the matrices
\[\bm{A}=\begin{pmatrix}
    \rho_u & 0\\
    0 & \rho_v
\end{pmatrix}, \bm{B}=\begin{pmatrix}
    1 & 0 & 0 & 0 & 0 & 0\\
    0 & 1 & 0 & 0 & 0 & 0
\end{pmatrix}, \bm{C}=\begin{pmatrix}
    1 & 0\\
    0 & 1
\end{pmatrix},\bm{D}=\begin{pmatrix}
    0 & 0 & 1 & 0 & 1 & 0\\
    0 & 0 & 0 & 1 & 0 & 1
\end{pmatrix}.\]
Since the formal spectral density approach does not include covariates, $\bm{s}_t$ is considered as an additive noise and the covariate is ``internalized''. The input process is thus defined by
\begin{equation}
    \forall t\geq 0,\quad \bm{u}_t:= \begin{pmatrix}
    \bm{\varepsilon}_t\\
    \bm{s}_t\\
    \bm{\eta}_t
\end{pmatrix}\sim \mathcal{N}(\bm{0},\bm{Q}), \quad \text{with} \quad
\bm{Q}=\begin{pmatrix}
    \bm{Q}_1 & \bm{0} & \bm{0}\\
    \bm{0}  & \bm{I}_2 & \bm{0}\\
    \bm{0} & \bm{0} & \bm{Q}_2
\end{pmatrix}.
\end{equation}
Removing the equal coefficients in the output spectral, the exhaustive summary writes
\begin{equation}
    \bm{\kappa}(\bm{\theta})=\begin{pmatrix}
    \rho_u\sigma^2_{\eta,u}+\rho_u\\
    -\rho_u^2\sigma^2_{\eta,u}-\rho_u^2-\sigma^2_{\eta,u}-\sigma^2_{\varepsilon,u}-1\\
    \rho_u\\
    -\rho_u^2-1\\
    \rho_v\sigma^2_{\eta,v}+\rho_v\\
    -\rho_v^2\sigma^2_{\eta,v}-\rho_v^2-\sigma^2_{\eta,v}-\sigma^2_{\varepsilon,v}-1\\
    \rho_v\\
    -\rho_v^2-1\\
    \end{pmatrix}.
\end{equation}
Now, we compute the derivative matrix 
\begin{align}
    & \dfrac{\partial \bm{\kappa}}{\partial \bm{\theta}} = \\&\begin{pmatrix}
    \sigma^2_{\eta,u}+1 & -2\rho_u\sigma^2_{\eta,u}-2\rho_u & 1 & -2\rho_u & 0 & 0 & 0 & 0\\
    0  & 0 & 0 & 0 & \sigma^2_{\eta,v}+1 & -2\rho_u\sigma^2_{\eta,u}-2\rho_u & 1 & -2\rho_v\\
    \rho_u & -\rho_u^2-1 & 0 & 0 & 0 & 0 & 0 & 0\\
    0 & 0 & 0 & 0 & \rho_v & -\rho_v^2-1 & 0 & 0\\
    0 & -1 & 0 & 0 & 0 & 0 & 0 & 0\\
    0 & 0 & 0 & 0 & 0 & -1 & 0 & 0
    \end{pmatrix} \nonumber
\end{align}
and we calculate the rank. As the derivative matrix has rank $r=6$, and there are six parameters, this model is not parameter redundant.

\subsection{Interaction between species}
\label{InteractionSpecies}

Inspired by the Multivariate AutoRegressive (MAR(1), also known as VAR(1)) model presented by \citet{Ives2003estimating} to estimate interaction between species, we consider the following state process
\begin{equation}
    \bm{x}_{t+1}=\begin{pmatrix}
    a_{11} & a_{12}\\
    a_{21} & a_{22}
\end{pmatrix}\bm{x}_t+\bm{\varepsilon}_t,\quad \bm{\varepsilon}_t\sim\mathcal{N}(\bm{0},\sigma^2_{\varepsilon}\bm{I}_2),
\end{equation}
where $\bm{x}_t$ is the log-transformed population abundance and the elements $a_{ij}$ give the effect of the (log-)abundance of species $j$ on the population growth rate of species $i$. In order to have a SSM, we assume that only the first species is observed, with some observation error, that is
\begin{equation}
    y_t=\begin{pmatrix}
    1 & 0
\end{pmatrix}\bm{x}_t+\eta_t,\quad \eta_t\sim \mathcal{N}(0,\sigma^2_{\eta}).
\end{equation}
The vector of parameters is $\bm{\theta}=(a_{11},a_{12},a_{21},a_{22},\sigma^2_{\eta}, \sigma^2_{\varepsilon})$. 
Removing the equal coefficients of the output spectral density, the exhaustive summary writes
\begin{equation}
   \bm{\kappa}(\bm{\theta})=   
   \begin{pmatrix}
    (a_{11}a_{22}-a_{12}a_{21})\sigma^2_{\eta}\\
    -(a_{11}+a_{22})(a_{11}a_{22}-a_{12}a_{21}+1)\sigma^2_{\eta}-a_{22}\sigma^2_{\varepsilon}\\
    \kappa_3(\bm{\theta})\\
    a_{11}a_{22}-a_{12}a_{21}\\
    -(a_{11}+a_{22})(a_{11}a_{22}-a_{12}a_{21}+1)\\
    (a_{11}a_{22}-a_{12}a_{21})^2+(a_{11}+a_{22})^2+1
\end{pmatrix}, 
\end{equation}
with $\kappa_3(\bm{\theta})=((a_{11}^2+1)a_{22}^2+(-2a_{12}a_{21}+2)a_{11}a_{22}+a_{12}^2a_{21}^2+a_{11}^2+1)\sigma^2_{\eta}+(a_{12}^2+a_{22}^2+1)\sigma^2_{\varepsilon}$. The derivative matrix $\partial \bm{\kappa}/\partial \bm{\theta}$ has rank $r=5$, but there are six parameters, this model is parameter redundant with deficiency $d=1$. Solving the equation $\bm{\alpha}^{\top} \partial \bm{\kappa}/\partial \bm{\theta}=\bm{0}$ gives 
\begin{equation}
    \bm{\alpha} = \begin{pmatrix}
    \dfrac{a_{22}}{\sigma^2_{\varepsilon}} \\
    -\dfrac{a_{12}^2-a_{22}^2+1}{2a_{12}\sigma^2_{\varepsilon}} \\
    -\dfrac{2a_{11}a_{22}a_{12}-a_{12}^2a_{21}-2a_{12}a_{22}^2+a_{21}a_{22}^2-a_{21}}{2a_{12}^2\sigma^2_{\varepsilon}} \\
    -\dfrac{a_{22}}{\sigma^2_{\varepsilon}} \\
    0 \\
    1
    \end{pmatrix}.
\end{equation}
The position of the zero indicates that the parameter $\sigma^2_{\eta}$ is individually identifiable. Solving the partial differential equation 
\begin{align}
    &\dfrac{a_{22}}{\sigma^2_{\varepsilon}} \dfrac{\partial f}{\partial a_{11}}-\dfrac{a_{12}^2-a_{22}^2+1}{2a_{12}\sigma^2_{\varepsilon}} \dfrac{\partial f}{\partial a_{12}}-\dfrac{2a_{11}a_{22}a_{12}-a_{12}^2a_{21}-2a_{12}a_{22}^2+a_{21}a_{22}^2-a_{21}}{2a_{12}^2\sigma^2_{\varepsilon}}\dfrac{\partial f}{\partial a_{21}} \nonumber\\
    & -\dfrac{a_{22}}{\sigma^2_{\varepsilon}}\dfrac{\partial f}{\partial a_{22}} + \dfrac{\partial f}{\partial \sigma^2_{\varepsilon}}=0
\end{align}
gives $\beta_1=a_{11}+a_{22}$, $\beta_2=\sigma^2_{\eta}$, $\beta_3=a_{22}\sigma^2_{\varepsilon}$, $\beta_4=(a_{12}^2-a_{22}a_{11}+1)/a_{22}$ and $\beta_5=a_{11}a_{22}-a_{12}a_{21}$ as the estimable parameters. As an illustration, the exhaustive summary can be expressed in terms of a smaller parameter vector. Indeed
\begin{equation}
    \bm{\kappa}(\bm{\beta})=\begin{pmatrix}
    \beta_2\beta_5\\
    -\beta_2(\beta_5+1)\beta_1-\beta_3\\
    (\beta_1^2+\beta_5^2+1)\beta_2+(\beta_1+\beta_4)\beta_3\\
    \beta_5\\
    -(\beta_5+1)\beta_1\\
    \beta_1^2+\beta_5^2+1
\end{pmatrix}.
\end{equation}
Supplementary Material \ref{InteractionBetweenSpecies} gives this SSM when all the species are observed. In this case, the model is identifiable.

\subsection{Summary of results}

Let us first mention the simplest example, the classic AR(1) model with an observation equation (example \ref{AR1}), which is the model of \citet{knape2008estimability} and the first model presented by \citet{AugerMethe2016state}, who highlighted persistent identifiability issues in simple SSMs. They showed with a simulation study that the distribution of the variance parameter estimates was bimodal, following earlier work by \citet{dennis2006estimating}.  \citet{valpine2005State} and \citet{knape2008estimability} showed that an overestimated measurement variance could be compensated by an underestimated process variance, and vice versa \citep[see also more recently][for a multispecies extension]{certain2018how}. This suggests that the two quantities cannot be estimated separately in the absence of repeated measures at a given time, as envisioned by \citet{dennis2010replicated}. However, the exhaustive summary presented in Section \ref{AR1} demonstrates that the AR(1) model is theoretically identifiable. This result, theoretical identifiability yet practical non-identifiability of variance and potentially other parameters seems to follow for  ecological SSMs.  Section \ref{AnimalMovement} presents an animal movement model with a covariate. \citet{AugerMethe2016state} observed that parameter estimates vary widely across polar bear individuals. But theoretical identifiability is here demonstrated by the symbolic differentiation approach. The MAR(1) model defined in Section \ref{InteractionSpecies} to estimate interactions between species is a SSM when an observation equation is considered. If some species are unobserved it is not identifiable. But Supplementary Material \ref{InteractionBetweenSpecies} indicates that the interaction model is identifiable when all the species are observed. The above results suggest that classic linear SSMs used in ecology are theoretically identifiable. The estimation issues encountered \citep{Ives2003estimating, dennis2006estimating, AugerMethe2016state, certain2018how} are therefore due to practical non-identifiability, plausibly generated by too short or not enough variable time series (even though these might appear quite long by ecological standards). 

Now, we investigate the reasons for which a linear SSM can be theoretically non-identifiable. The population dynamics model described by Section \ref{PopulationDynamics} is found to be non-identifiable due to the two parameters only ever appearing as a product, a very classical source of non-identifiability issues. We need to rewrite the product as a unique parameter in order to estimate it: as expected, we confirm the result of \citet{cole2016parameter} on the same model with the new exhaustive summary. A second non-identifiability source explored in Section \ref{AR1} is the interaction between observation and ecological processes. Indeed, even the classic AR(1) model with observation error becomes non-identifiable when some covariance between the observation and process error is added---this is rarely done in practice fortunately, but could very well happen in the field (e.g. through bad weather conditions that decrease both population growth and the likelihood to see the organisms). Another, less obvious reason for non-identifiability is the degradation of a model. As mentioned above, Section \ref{InteractionSpecies} presents a two-species interaction model when only the first species is observed. When this loss of information occurs, the model becomes non-identifiable. 

\section{Discussion}
\label{Discussion}

In this paper, we construct an exhaustive summary for linear and stochastic state-space models that accounts for both mean process and variance parameters, building on the spectral representation of the stochastic process. This new exhaustive summary is then applied to three stochastic ecological systems. The analyses suggest that, in spite of identifiability issues being found in various SSMs observed for different lengths of time, classic linear SSMs frequently used by ecologists are theoretically identifiable (once obvious sources of unidentifiability such as  parameters appearing as products everywhere have been removed). Therefore, issues that have been pinpointed in such models are of a practical nature, due to data limitations. The one exception being models that observe only one compartment of a multi-compartment system. Although many classical ecology surveys try to observe all components of the ecosystem that they wish to model \citep{Ives2003estimating}, it is not always so \citep[e.g.,][]{ives2008high}. 

In addition to helping understand when we can theoretically identify variance parameters, can the new exhaustive summary provide different answers about mean process parameter identifiability than previous exhaustive summaries based on expectations of the stochastic process \citep{cole2016parameter, cole2020parameter}? Since the ecological model examples presented above are identifiable or not both in their stochastic and deterministic versions, to answer this question, Supplementary Material \ref{DifferencesES} presents three versions of a compartmental model which is more suited to the task: deterministic, stochastic without variance parameter and stochastic with variance parameter. This case study shows that the model is identifiable when no variance parameter is involved, but parameter redundant if the variance parameter needs to be estimated (with an irreductible parameter combination including both variance and mean process parameter). This demonstrates that in a number of cases, the output spectral density exhaustive summary will yield different -- and more exact -- conclusions than previously used exhaustive summaries for linear SSMs.

Previous research has also considered constructing exhaustive summaries by adding variance of observed variables in addition to expectations. This approach is outlined in the supplementary of \citet{cole2016parameter} and then applied in \citet{polansky2020improving,polansky2023combining} as well as \citet{aldrin2021caveats}. Except in the latter work, the technique has been applied to the variance of the individual variables $\mathbb{V}[y_{i,t}]$ rather than the full covariance matrix, and in \citet{aldrin2021caveats} is it applied to the \textit{instantaneous} variance-covariance matrix $\mathbb{V}[\bm{y}_{t}]$. Because the output spectral density, which provides an exact exhaustive summary, is the Fourier transform of the autocovariance $\bm{\Gamma}^h_y=\text{Cov}(\bm{y}_{t+h},\bm{y}_{t})$ at the stationary state, we know that the autocovariance can also provide exhaustive summaries.
Therefore, by validating the use of the output spectral density we also show that simpler variance- or covariance-based summaries (without the lagged terms of the autocovariance, and computed out of equilibrium) do not necessarily contain the full information to represent the model in general, although they might in some cases. Supplementary Material \ref{autocovariance} builds an exhaustive summary based on the autocovariance and explores an example.

We discuss now some limitations and potential perspectives of this work. We first discuss a rather practical limitation. We used the symbolic differentiation approach to determine the (non-)identifiability of the model. As already noticed by \citet{cole2010determining} and \citet{cole2016parameter}, when the exhaustive summary is structurally much more complex than the ones considered here, the computer may run out of memory trying to calculate the rank of the derivative matrix. For instance, in the empirical examples that we have used and using Maple code, the rank could not be computed when the interactions between four (or more) species are estimated. In order to investigate the identifiability in this case, some numerical approach could be required \citep[][Chapter 4]{cole2020parameter}. In a parametric statistical model, a classical numerical approach consists in checking the profile of the log-likelihood and evaluate the rank of the exhaustive summary numerically, which can be tricky due to sensitivity to numerical precision. A compromise between the numerical and the symbolic approach put forward by \citet{cole2020parameter} consists in a hybrid symbolic-numerical method, which finds the derivative matrix symbolically and then calculates the rank numerically for multiple sets of parameter values. This will likely be a reasonable approach to use the exhaustive summary presented above for much more complex SSMs. 

Throughout this paper, we made a strong assumption of stationarity. Similarly to \citet{glover1974parametrizations}, we assume that the output process has reached steady state when the practitioner starts the observations. In practice, there are many ecological (or other) datasets where this might be untrue: for instance, one might observe a biological invasion in real time, where a species will replace another one, or the system might be under the influence of a non-stationary climate variable, etc. Perturbations of the dynamics could nevertheless help to better estimate the parameters of the time series \citep{cao2017HumanGut}. Moreover, since the input process is fully stochastic, the model is supposed to be centered and covariates are not allowed (unless these are modelled as endogeneous to the model, e.g., the covariate process is specified too). These two limitations could be overcome by  using an exhaustive summary built on the Kalman filter, which can be used for the linear SSMs tackled in this paper. Supplementary Material \ref{KalmanFilterSection} presents how to investigate the identifiability of the SSM with the Kalman filter. We note that the derived exhaustive summary may be complex to use for fully symbolic differentiation; some hybrid symbolic-numerical method could be useful here as well.

Linear SSMs are rather popular in ecology and beyond. However, many processes can require nonlinear models \citep{AugerMethe2021Aguide}, and this assumption of linearity limits the exploration of the identifiability of numerous models, at least using a transfer function approach. The spectral density approach, which builds on the transfer function, seems less likely to be useful for nonlinear processes, and again an exhaustive summary could be derived from the Kalman filter, but this time in its extended formulation \citep{julier2004unscented}, for nonlinear processes. This could provide a natural extension to the present work. 

\section*{Acknowledgments}

This research was supported by grant ANR-20-CE45-0004 to FB. We thank Stéphane Robin and Bixuan Liu for discussions on identifiability in VAR(1) models. 

\bibliographystyle{apalike}
\bibliography{mabiblio}

\newpage

\renewcommand{\theequation}{S\arabic{equation}}
\renewcommand{\thesection}{S\arabic{section}}
\renewcommand{\thefigure}{S\arabic{figure}}

\setcounter{equation}{0}
\setcounter{section}{0}
\setcounter{figure}{0}

\section{Supplementary Materials}

\subsection{Derivations for the transfer function and output spectral density}
\label{derivations}

A linear discrete in time SSM is defined by the equations 
\begin{equation}
    \bm{x}_{t+1}=\bm{Ax}_t+\bm{Bu}_t \quad \text{and} \quad \bm{y}_t=\bm{Cx}_t+\bm{Du}_t,
\end{equation}
where the input sequence $(\bm{u}_t)$ is assumed to be a white noise process, such that for all $t \geq 0$, $\bm{u}_t\sim \mathcal{N}(\bm{0}, \bm{Q})$. To find the transfer function we use the discrete-time Fourier transform. Assuming that $\bm{x}_0=\bm{0}$, we have
\begin{align}
   & \mathcal{F}\left\{\bm{x}_{t+1}\right\}(i\omega)=\mathcal{F}\left\{\bm{A}\bm{x}_t+\bm{B}\bm{u}_t\right\}(i\omega)\nonumber\\
   \Rightarrow  \quad & \exp(i\omega)\bm{X}(i\omega)=\bm{AX}(i\omega)+\bm{BU}(i\omega) \nonumber\\
   \Rightarrow   \quad & \bm{X}(i\omega) =\left(\exp(i\omega)\bm{I}-\bm{A}\right)^{-1}\bm{BU}(i\omega)
\end{align}
where $\bm{X}(i\omega)$ and $\bm{U}(i\omega)$ define respectively the discrete-time Fourier transform of $(\bm{x}_t)$ and $(\bm{u}_t)$ at point $i\omega$. The discrete-time Fourier transform of $(\bm{y}_t)$ thus becomes
\begin{align}
    \bm{Y}(i\omega)&=\bm{C}\left(\exp(i\omega)\bm{I}-\bm{A}\right)^{-1}\bm{BU}(i\omega)+\bm{DU}(i\omega) \nonumber\\
    & =\left[\bm{C}\left(\exp(i\omega)\bm{I}-\bm{A}\right)^{-1}\bm{B}+\bm{D}\right]\bm{U}(i\omega) \nonumber\\
    &= \bm{H}_{s}(\bm{\theta})\bm{U}(i\omega)
\end{align}
where $\bm{H}_{s}(\bm{\theta}):=\bm{C}\left(s\bm{I}-\bm{A}\right)^{-1}\bm{B}+\bm{D}$ defines the transfer function with  $s=\exp(i\omega)$. The output spectral density writes
\begin{align}
    \bm{S}(\bm{\theta})&=\mathbb{E}\left[\bm{Y}(i\omega)\bm{Y}(-i\omega)^{\top}\right] \nonumber\\
    &=\bm{H}_{s}(\bm{\theta})\mathbb{E}\left[\bm{U}(i\omega)\bm{U}(-i\omega)^{\top}\right]\bm{H}_{1/s}^{\top}(\bm{\theta}) \nonumber\\
    &=\bm{H}_{s}(\bm{\theta})\bm{Q}\bm{H}_{1/s}^{\top}(\bm{\theta})
\end{align}

\subsection{On the continuous counterpart}
\label{ContinuousCounterpart}
\subsubsection{Derivations for the transfer function and output spectral density}

A linear continuous in time SSM  is defined by the equations 
\begin{equation}
    \dfrac{\partial \bm{x}(t)}{\partial t}=\bm{Ax}(t)+\bm{Bu}(t) \quad \text{and} \quad \bm{y}(t)=\bm{Cx}(t)+\bm{Du}(t),
\end{equation}
where the input $\bm{u}$ is assumed to be a white noise process with $\mathbb{E}[\bm{u}(t_1)\bm{u}(t_2)^{\top}]=\bm{Q}\delta(t_1-t_2)$. To find the transfer function we use the continuous-time Fourier transform. Assuming that $\bm{x}(0)=\bm{0}$, we have
\begin{align}
   & \mathcal{F}\left\{\dfrac{\partial \bm{x}(t)}{\partial t}\right\}(i\omega)=\mathcal{F}\left\{\bm{Ax}(t)+\bm{Bu}(t)\right\}(i\omega) \nonumber\\
   \Rightarrow  \quad & (i\omega)\bm{X}(i\omega)=\bm{AX}(i\omega)+\bm{BU}(i\omega) \nonumber\\
   \Rightarrow   \quad & \bm{X}(i\omega) =\left((i\omega)\bm{I}-\bm{A}\right)^{-1}\bm{BU}(i\omega)
\end{align}
where $\bm{X}(i\omega)$ and $\bm{U}(i\omega)$ define respectively the continuous-time Fourier transform of $\bm{x}$ and $\bm{u}$ at point $(i\omega)$. The continuous-time Fourier transform of $\bm{y}$ thus becomes
\begin{align}
    \bm{Y}(i\omega)&=\bm{C}\left((i\omega)\bm{I}-\bm{A}\right)^{-1}\bm{BU}(i\omega)+\bm{DU}(i\omega) \nonumber\\
    & =\left[\bm{C}\left((i\omega)\bm{I}-\bm{A}\right)^{-1}\bm{B}+\bm{D}\right]\bm{U}(i\omega) \nonumber\\
    &= \bm{H}_{s}(\bm{\theta})\bm{U}(i\omega)
\end{align}
where $\bm{H}_{s}(\bm{\theta}):=\bm{C}\left(s\bm{I}-\bm{A}\right)^{-1}\bm{B}+\bm{D}$ defines the transfer function with  $s=(i\omega)$. To assure the existence of the transfer function, $s$ must be different from the eigenvalues of $\bm{A}$. The output spectral density writes
\begin{align}
    \bm{S}(\bm{\theta})&=\mathbb{E}\left[\bm{Y}(i\omega)\bm{Y}(-i\omega)^{\top}\right] \nonumber\\
    &=\bm{H}_{s}(\bm{\theta})\mathbb{E}\left[\bm{U}(i\omega)\bm{U}(-i\omega)^{\top}\right]\bm{H}_{-s}^{\top}(\bm{\theta}) \nonumber\\
    &=\bm{H}_{s}(\bm{\theta})\bm{Q}\bm{H}_{-s}^{\top}(\bm{\theta}).
\end{align}

\subsubsection{Theoretical result}
Following the same lines as the paper, we obtain 

\begin{theorem}
    For stochastic linear-time-invariant continuous-time SSMs, the components of the output spectral density have the form
    \begin{equation}
        [\bm{S}(\bm{\theta})]_{ij}=\dfrac{\sum_{n=0}^{2N}a_n(\bm{\theta})s^{n}}{\sum_{n=0}^{2N}b_n(\bm{\theta})s^{n}},
    \end{equation}
    and the output spectral density exhaustive summary $\bm{\kappa}(\bm{\theta})$ consists of the non-constant coefficients $a_n(\bm{\theta})$ and $b_n(\bm{\theta})$ of the powers of $s$ in the numerators and denominators of $\bm{S}(\bm{\theta})$.
\end{theorem}

\subsubsection{Example}

Here, we consider the stochastic counterpart of the two-compartment model presented by \citet{bellman1970on}.  The state equation writes
\begin{equation}
    \dfrac{\partial \bm{x}(t)}{\partial t} = \begin{pmatrix}
    -\theta_{0,1}-\theta_{2,1} & \theta_{1,2} \\
    \theta_{2,1} & -\theta_{0,2}-\theta_{1,2}
\end{pmatrix}\bm{x}(t)+\begin{pmatrix}
    1\\
    0
\end{pmatrix}u(t)
\end{equation}
where the input $u$ is assumed to be a white noise process with $\mathbb{E}[u(t_1)u(t_2)]=\sigma^2\delta(t_1-t_2)$. The observation equation is given by
\begin{equation}
    y(t)=\begin{pmatrix}
    1 & 0
\end{pmatrix}\bm{x}(t).
\end{equation}
The vector of parameters is $\bm{\theta}=(\sigma^2, \theta_{0,1},\theta_{0,2},\theta_{1,2},\theta_{2,1})$. Removing the constant coefficients, the exhaustive summary writes
\begin{equation}
    \bm{\kappa}(\bm{\theta})=\begin{pmatrix}
    (\theta_{0,2}+\theta_{1,2})^2\sigma^2\\
    -\sigma^2\\
    ((\theta_{0,2}+\theta_{1,2})\theta_{0,1}+\theta_{0,2}\theta_{2,1})^2\\
    -(\theta_{0,1}+\theta_{2,1})^2-(\theta_{0,2}+\theta_{1,2})^2-2\theta_{1,2}\theta_{2,1}
\end{pmatrix}.
\end{equation}
The derivative matrix has rank $r=4$, but there are five parameters, this model is parameter redundant with deficiency $d=1$. Solving the equation $\bm{\alpha}^{\top} (\partial \bm{\kappa}/\partial \bm{\theta})=\bm{0}$ gives 
\begin{equation}
    \bm{\alpha}^{\top}=\begin{pmatrix}
    0 & -1 & \dfrac{\theta_{1,2}}{\theta_{2,1}} & -\dfrac{\theta_{1,2}}{\theta_{2,1}} & 1
\end{pmatrix}.
\end{equation}
Solving the partial differential equation 
\begin{equation}
    -\dfrac{\partial f}{\partial\theta_{0,1}}+\dfrac{\theta_{1,2}}{\theta_{2,1}}\dfrac{\partial f}{\partial\theta_{0,2}}-\dfrac{\theta_{1,2}}{\theta_{2,1}}\dfrac{\partial f}{\partial\theta_{1,2}}+\dfrac{\partial f}{\partial\theta_{2,1}}=0
\end{equation}
gives $\beta_1=\sigma^2$, $\beta_2=\theta_{0,1}+\theta_{2,1}$, $\beta_3=\theta_{1,2}\theta_{2,1}$ and $\beta_4=\theta_{0,2}+\theta_{1,2}$ as estimable parameters.

\subsection{Interaction between species with full-system observations}
\label{InteractionBetweenSpecies}

Inspired by the Multivariate AutoRegressive Model presented by \citet{Ives2003estimating} to estimate interaction between species, we consider the following state process
\begin{equation}
    \bm{x}_{t+1}=\begin{pmatrix}
    a_{11} & a_{12}\\
    a_{21} & a_{22}
\end{pmatrix}\bm{x}_t+\bm{\varepsilon}_t,\quad \bm{\varepsilon}_t\sim\mathcal{N}(\bm{0},\sigma^2_{\varepsilon}\bm{I}_2),
\end{equation}
where $\bm{x}_t$ is the log-transformed population abundance and the elements $a_{ij}$ give the effect of the (log-)abundance of species $j$ on the population growth rate of species $i$. In order to have a state-space form, we now assume that all the species are observed with error, with the same observation variance for all species, that is
\begin{equation}
    \bm{y}_t=\bm{x}_t+\bm{\eta}_t,\quad \bm{\eta}_t\sim \mathcal{N}(\bm{0},\sigma^2_{\eta}\bm{I}_2).
\end{equation}
The vector of parameters is $\bm{\theta}=(a_{11},a_{12},a_{21},a_{22},\sigma^2_{\eta}, \sigma^2_{\varepsilon})$. Removing the equal coefficients of the output spectral density, the exhaustive summary writes
\begin{align}
   & \bm{\kappa}(\bm{\theta})=   \\
   &\begin{pmatrix}
    (a_{11}a_{22}-a_{12}a_{21})\sigma^2_{\eta}\\
    -(a_{11}+a_{22})(a_{11}a_{22}-a_{12}a_{21}+1)\sigma^2_{\eta}-a_{22}\sigma^2_{\varepsilon}\\
    ((a_{11}^2+1)a_{22}^2+(-2a_{12}a_{21}+2)a_{11}a_{22}+a_{12}^2a_{21}^2+a_{11}^2+1)\sigma^2_{\eta}+(a_{12}^2+a_{22}^2+1)\sigma^2_{\varepsilon}\\
    a_{11}a_{22}-a_{12}a_{21}\\
    -(a_{11}+a_{22})(a_{11}a_{22}-a_{12}a_{21}+1)\\
    (a_{11}a_{22}-a_{12}a_{21})^2+(a_{11}+a_{22})^2+1\\
    a_{12}\sigma^2_{\varepsilon}\\
    -(a_{11}a_{22}+a_{12}a_{21})\sigma^2_{\epsilon}\\
    -(a_{11}+a_{22})(a_{11}a_{22}-a_{12}a_{21}+1)\sigma^2_{\eta}-a_{11}\sigma^2_{\varepsilon}\\
    ((a_{22}^2+1)a_{11}^2+(-2a_{12}a_{21}+2)a_{11}a_{22}+a_{12}^2a_{21}^2+a_{22}^2+1)\sigma^2_{\eta}+(a_{11}^2+a_{21}^2+1)\sigma^2_{\varepsilon}
\end{pmatrix}. \nonumber
\end{align}
We compute the derivative matrix and we calculate the rank. As the derivative matrix has rank $r=6$, and there are six parameters, this model is not parameter redundant.

\subsection{Differences between exhaustive summaries in absence and presence of noise}
\label{DifferencesES}

This section is dedicated to the presentation of three exhaustive summaries designed to investigate the (non)identifiability of the SSMs. We consider the same dynamical system as the paper, that is
\begin{equation}
    \bm{x}_{t+1}=\bm{Ax}_t+\bm{Bu}_t
\end{equation}
and
\begin{equation}
    \bm{y}_t=\bm{Cx}_t+\bm{Du}_t.
\end{equation}
To illustrate the differences between these exhaustive summaries, we study a model given by \citet[][Section 7.1.2]{cole2020parameter}. The matrices of this compartmental model are
\begin{equation*}
    \bm{A}=\begin{pmatrix}
        -\theta_{21} & 0 & 0\\
        \theta_{21} & -\theta_{32}-\theta_{02} & 0\\
        0 & \theta_{32} & 0
    \end{pmatrix}, \bm{B}=\begin{pmatrix}
        1\\
        0\\
        0
    \end{pmatrix}, \bm{C}=\begin{pmatrix}
        0 & 0 & 1
    \end{pmatrix} \quad \text{and} \quad \bm{D}=\bm{0}.
\end{equation*}
The first two exhaustive summaries are respectively given by Chapter 7 and Chapter 8 of \citet{cole2020parameter}.

\subsubsection{Exhaustive summary based on the transfer function}
Here, the sequence $(u_t)$ is deterministic. The transfer function writes 
\begin{align}
    \bm{H}_1(\bm{\theta})&=\bm{C}\left(s\bm{I}-\bm{A}\right)^{-1}\bm{B} \nonumber\\
    &= \dfrac{\theta_{21}\theta_{32}}{s^3+(\theta_{21}+\theta_{32}+\theta_{02})s^2+\theta_{21}(\theta_{32}+\theta_{02})s}
\end{align}
In the same way as Section \ref{MainTheorem}, the exhaustive summary $\bm{\kappa}_1$ encapsulates all the coefficients of all the polynomials of the transfer function. Indeed, 
\begin{equation}
    \bm{\kappa}_1(\bm{\theta})=\begin{pmatrix}
        \theta_{21}\theta_{32}\\
        \theta_{21}+\theta_{32}+\theta_{02}\\
        \theta_{21}(\theta_{32}+\theta_{02})
    \end{pmatrix}.
\end{equation}
The associated derivative matrix has rank $r=3$ and there are three parameters, this model is not parameter redundant.

\subsubsection{Exhaustive summary based on the (almost)  transfer function}
Here, the sequence $(u_t)$ is stochastic, but there is no variance parameter to estimate, i.e. $u_t\sim\mathcal{N}(0,\sigma^2)$ with $\sigma^2$ known. The identifiability is checked through the expectation of the model. We have 
\begin{equation}
    \mathbb{E}[\bm{x}_{t+1}]=\bm{A} \mathbb{E}[\bm{x}_{t}] \quad \text{and} \quad \mathbb{E}[\bm{y}_{t}]=\bm{C} \mathbb{E}[\bm{x}_{t}].
\end{equation}
Now, the transfer function writes
\begin{align}
    \bm{H}_2(\bm{\theta})&=\bm{C}\left(s\bm{I}-\bm{A}\right)^{-1}\bm{A}\bm{x}_0 \nonumber\\
    &=\dfrac{\theta_{32}x_{02}s+\theta_{32}(\theta_{21}x_{01}+\theta_{21}x_{02})}{s^2+(\theta_{21}+\theta_{32}+\theta_{02})s+\theta_{21}(\theta_{32}+\theta_{02})}
\end{align}
where $\bm{x}_0=(x_{01},x_{02},x_{03})$ is the vector of initial values. The derivations of the transfer function are given by \citet{cole2016parameter}. The exhaustive summary now writes
\begin{equation}
    \bm{\kappa}_2(\bm{\theta})=\begin{pmatrix}
        \theta_{32}(\theta_{21}x_{01}+\theta_{21}x_{02})\\
        \theta_{32}x_{02}\\
        \theta_{21}(\theta_{32}+\theta_{02})\\
        \theta_{21}+\theta_{32}+\theta_{02}
    \end{pmatrix}.
\end{equation}
The associated derivative matrix has rank $r=3$ and there are three parameters, this model is not parameter redundant.

\subsubsection{Exhaustive summary based on the output spectral density}
Here, the sequence $(u_t)$ is stochastic, and the variance parameter $\sigma^2$ needs to be estimated. Following Section \ref{Definition}, we get
\begin{equation}
    \bm{\kappa}_3(\bm{\theta})=\begin{pmatrix}
        \theta_{21}^2\theta_{32}^2\sigma^2\\
        \theta_{21}(\theta_{32}+\theta_{02})\\
        (\theta_{21}+\theta_{32}+\theta_{02})(\theta_{21}(\theta_{32}+\theta_{02})+1)\\
        \theta_{21}^2(\theta_{32}+\theta_{02})^2+(\theta_{21}+\theta_{32}+\theta_{02})^2+1
    \end{pmatrix}.
\end{equation}
The derivative matrix has rank $r=3$, but there are four parameters, this model is parameter redundant with deficiency $d=1$. Solving the equation $\bm{\alpha}^{\top} (\partial \bm{\kappa}_3/\partial \bm{\theta})=\bm{0}$ gives 
\begin{equation}
    \bm{\alpha}^{\top}=\begin{pmatrix}
    -\dfrac{2\sigma^2}{\theta_{32}} & -1 & 0 & 1
\end{pmatrix}.
\end{equation}
Solving the partial differential equation 
\begin{equation}
    -\dfrac{2\sigma^2}{\theta_{32}}\dfrac{\partial f}{\partial \sigma^2}-\dfrac{\partial f}{\partial\theta_{02}}+\dfrac{\partial f}{\partial\theta_{32}}=0
\end{equation}
gives $\beta_1=\theta_{21}$, $\beta_2=\theta_{32}^2\sigma^2$ and $\beta_3=\theta_{02}+\theta_{32}$ as estimable parameters.

\subsection{Exhaustive summary based on the autocovariance}
\label{autocovariance}

In this appendix, we consider the stochastic SSM given in Section \ref{Definition}
\begin{equation}
    \bm{x}_{t+1}=\bm{Ax}_t+\bm{Bu}_t \quad \text{and} \quad \bm{y}_t=\bm{Cx}_t+\bm{Du}_t,
\end{equation}
where the input sequence $(\bm{u}_t)$ is assumed to be a white noise process, such that for all $t>0$, $\bm{u}_t\sim\mathcal{N}(0,\bm{Q})$. We want to find an explicit formula for the successive autocovariance terms defined by $\bm{\Gamma}_y^h=\text{Cov}(\bm{y}_{t+h}, \bm{y}_t)$. Inspired by \citet[][Section 2.1.4 and Section 18.2.1]{lutkepohl2006new}, we first compute $\bm{\Gamma}_x^h=\text{Cov}(\bm{x}_{t+h}, \bm{x}_t)$.

\subsubsection{Derivations for the autocovariance of the latent process}

We have
\begin{align}
    \bm{\Gamma}_x^h & = \text{Cov}(\bm{x}_{t+h}, \bm{x}_{t}) \nonumber\\
     & = \bm{A}\text{Cov}(\bm{x}_{t+h-1}, \bm{x}_{t})+\text{Cov}(\bm{B}\bm{u}_{t+h-1}, \bm{x}_{t}) \nonumber\\
     &= \bm{A}\bm{\Gamma}_x^{h-1}+\text{Cov}(\bm{B}\bm{u}_{t+h-1}, \bm{x}_{t}).
\end{align}
Here $\text{Cov}(\bm{B}\bm{u}_{t+h-1}, \bm{x}_{t})=0$ if $h>0$ since $\bm{x}_t$ depends only on time up to $t-1$ while $\bm{u}_{t+h-1}$ is in the future (we call this argument ($\star$)). 
If $h=0$, we have 
\begin{align}
    \text{Cov}(\bm{B}\bm{u}_{t-1}, \bm{x}_{t}) &= \bm{B}\text{Cov}(\bm{u}_{t-1}, \bm{Ax}_{t-1}+\bm{Bu}_{t-1}) \nonumber\\
    &=\underbrace{\bm{B}\text{Cov}(\bm{u}_{t-1}, \bm{Ax}_{t-1})}_{=0\,\text{from}\,(\star)}+\bm{B}\text{Cov}(\bm{u}_{t-1}, \bm{u}_{t-1})\bm{B}^{\top} \nonumber\\
    =&\bm{BQB}^{\top},
\end{align}
and thus
\begin{equation}
    \bm{\Gamma}_x^h=\left\{\begin{matrix}
        \bm{A}\bm{\Gamma}_x^{-1}+\bm{BQB}^{\top},\quad h=0\\
        \bm{A}\bm{\Gamma}_x^{h-1},\quad h>0
    \end{matrix}\right.
\end{equation}
where $\bm{A}\bm{\Gamma}_x^{h-1}=\bm{A}^h\bm{\Gamma}_x^{0}$. Moreover, thanks to the previous derivations, we have $\bm{\Gamma}_x^0= \bm{A\Gamma}_x^{-1}+\bm{BQB}^{\top}= \bm{A}(\bm{\Gamma}_x^1)^{\top}+\bm{BQB}^{\top}$ and $\bm{\Gamma}_x^1=\bm{A\Gamma}^0_x$. We thus obtain 
\begin{equation}
    \bm{\Gamma}_x^0= \bm{A}\bm{\Gamma}_x^{0}\bm{A}^{\top}+\bm{BQB}^{\top}.
\end{equation}
To solve this equation we need to use the vec operator
\begin{align}
    &\text{vec}[\bm{\Gamma}_x^0]=(\bm{A}\otimes\bm{A})\text{vec}[\bm{\Gamma}_x^0]+\text{vec}[\bm{BQB}^{\top}] \nonumber\\
    \Rightarrow & \text{vec}[\bm{\Gamma}_x^0] = (\bm{I}_{N^2}-\bm{A}\otimes\bm{A})\text{vec}[\bm{BQB}^{\top}],
\end{align}
where $\otimes$ is the Kronecker product.

\subsubsection{Derivations when state and observation processes are independent}

Here, the covariance matrix $\bm{Q}$ of the white noise $(\bm{u}_t)$ is bloc-diagonal. We have
\begin{align}
    \bm{\Gamma}_y^h&=\text{Cov}(\bm{y}_{t+h}, \bm{y}_{t}) \nonumber\\
    &=\text{Cov}(\bm{Cx}_{t+h}+\bm{Du}_{t+h}, \bm{Cx}_{t}+\bm{Du}_{t})\nonumber\\
    &=\text{Cov}(\bm{Cx}_{t+h}, \bm{Cx}_{t}) + \text{Cov}(\bm{Cx}_{t+h},\bm{Du}_{t})+ \nonumber\\
    & \qquad \qquad \text{Cov}(\bm{Du}_{t+h}, \bm{Cx}_{t})+\text{Cov}(\bm{Du}_{t+h}\bm{Du}_{t})
\end{align}
Terms $\text{Cov}(\bm{Cx}_{t+h},\bm{Du}_{t})=\text{Cov}(\bm{Du}_{t+h}, \bm{Cx}_{t})=0$ since the state and observation process are independent. The autocovariance thus writes
\begin{equation}
    \bm{\Gamma}_y^h=\left\{\begin{matrix}
        \bm{C}\bm{\Gamma}_x^{0}\bm{C}^{\top}+\bm{DQD}^{\top},\quad h=0\\
        \bm{CA}^h\bm{\Gamma}_x^{0}\bm{C}^{\top},\quad h>0
    \end{matrix}\right.
\end{equation}


\subsubsection{Derivations when state and observation processes are not independent}

Now, since the matrix $\bm{Q}$ if of the form 
\begin{equation}
    \bm{Q}=\begin{pmatrix}
        \bm{Q}_1 & \bm{Q}_{12}\\
        \bm{Q}_{21} & \bm{Q}_2
    \end{pmatrix}
\end{equation}
the state and observation processes are not independent. We return to
\begin{align}
    \bm{\Gamma}_y^h&=\text{Cov}(\bm{Cx}_{t+h}, \bm{Cx}_{t}) + \text{Cov}(\bm{Cx}_{t+h},\bm{Du}_{t})+ \nonumber\\
    & \qquad \qquad \text{Cov}(\bm{Du}_{t+h}, \bm{Cx}_{t})+\text{Cov}(\bm{Du}_{t+h}\bm{Du}_{t})
\end{align}
Here $\text{Cov}(\bm{Du}_{t+h}, \bm{Cx}_{t})=0$ from $(\star)$, but we have to find an expression for $\text{Cov}(\bm{Cx}_{t+h},\bm{Du}_{t})$. By induction, we can proof that $\text{Cov}(\bm{x}_{t+h},\bm{Du}_{t})=\bm{A}^{h-1}\bm{BQD}^{\top}$. Indeed, for $h>1$ we have
\begin{align}
    \text{Cov}(\bm{x}_{t+h+1},\bm{Du}_{t}) &= \text{Cov}(\bm{Ax}_{t+h}+\bm{Bu}_{t+h},\bm{Du}_{t}) \nonumber\\
    &=\text{Cov}(\bm{Ax}_{t+h}, \bm{Du}_t)+\text{Cov}(\bm{Bu}_{t+h},\bm{Du}_t) \nonumber\\
    &=\bm{A}\text{Cov}(\bm{x}_{t+h}, \bm{Du}_t) \nonumber\\
    &=\bm{A}^{h}\bm{BQD}^{\top}
\end{align}
and for $h=1$, we start with
\begin{align}
    \text{Cov}(\bm{x}_{t+1},\bm{Du}_{t}) &= \text{Cov}(\bm{Ax}_t+\bm{Bu}_t,\bm{Du}_{t}) \nonumber\\
    &= \underbrace{\text{Cov}(\bm{Ax}_{t}, \bm{Du}_t)}_{=0\,\text{from}\,(\star)}+\text{Cov}(\bm{Bu}_t,\bm{Du}_t) \nonumber\\
    &=\bm{BQD}^{\top}
\end{align}
We thus obtain $\text{Cov}(\bm{Cx}_{t+h},\bm{Du}_{t})=\bm{CA}^{h-1}\bm{BQD}^{\top}$. Finally
\begin{equation}
\label{Gamma0y}
    \bm{\Gamma}_y^h=\left\{\begin{matrix}
        \bm{C}\bm{\Gamma}_x^{0}\bm{C}^{\top}+\bm{DQD}^{\top},\quad h=0\\
        \bm{CA}^h\bm{\Gamma}_x^{0}\bm{C}^{\top}+\bm{CA}^{h-1}\bm{BQD}^{\top},\quad h>0.
    \end{matrix}\right.
\end{equation}

\subsubsection{Construction of the exhaustive summary}

To capture all the parameter combinations, the exhaustive summary is given by the infinite concatenation of the $\text{vec}[\bm{\Gamma}^h_y]$'s, that is
\begin{equation}
    \bm{\kappa}=\begin{pmatrix}
        \text{vec}[\bm{\Gamma}^0_y]\\
        \text{vec}[\bm{\Gamma}^1_y]\\
        \text{vec}[\bm{\Gamma}^2_y]\\
        \vdots
    \end{pmatrix}.
\end{equation}
Since the length of the exhaustive summary is infinite, we need to find a way to determine the rank of the derivative matrix. A suitable approach can be found based on the arguments of \citet[][Option III]{cole2016parameter} or \citet[][Section 7.2.3.2]{cole2020parameter}. Indeed, if we consider the exhaustive summary $\bm{\kappa}_h$, for the first $h$ time lags, which is a function of all the parameters $\bm{\theta}$, the exhaustive summary can be extended to $\bm{\kappa}^{\top}=(\bm{\kappa}_h^{\top},\dots)$ with no extra parameter. Now, suppose the derivative matrix $\partial \bm{\kappa}_h/\partial \bm{\theta}$ is full rank, then $\partial \bm{\kappa}/\partial \bm{\theta}=[\partial \bm{\kappa}_h/\partial \bm{\theta}, \dots]$ is also full rank. However, if the derivative matrix $\partial \bm{\kappa}_h/\partial \bm{\theta}$ is not full rank, we need to recursively verify if the extra columns of the derivative matrix (corresponding to the derivatives of the extra terms of the exhaustive summary) still respect the null  linear combination of the rows. This alternative approach is illustrated with the autoregressive model of order one (AR(1)).

\subsubsection{Example with the AR(1) model}
\label{AR1Autocovariance}

We consider the  AR(1) model
\begin{equation}
  x_{t+1}=\rho x_t +\varepsilon_t, \quad \varepsilon_t \sim \mathcal{N}(0,\sigma^2_{\varepsilon})  
\end{equation}
with an observation equation
\begin{equation}
    y_t=x_t+\eta_t,\quad \eta_t \sim \mathcal{N}(0,\sigma^2_{\eta}).
\end{equation}
We thus have $A=\rho$, $B=(1\,\,0)$, $C=1$, $D=(0\,\,1)$ and
\begin{equation}
    \bm{Q}=\begin{pmatrix}
    \sigma^2_{\varepsilon} & 0\\
    0 & \sigma^2_{\eta}
\end{pmatrix}.
\end{equation}
If we compute $\text{vec}[\bm{\Gamma}^0_y]$, we have 
\begin{align}
    \text{vec}[\bm{\Gamma}^0_y] &= (\bm{C}\otimes \bm{C})\text{vec}[\bm{\Gamma}_x^0]+\text{vec}[\bm{DQD}^{\top}] \nonumber\\
    &=(\bm{C}\otimes \bm{C})(\bm{I}_{N^2}-\bm{A}\otimes\bm{A})\text{vec}[\bm{BQB}^{\top}]+\text{vec}[\bm{DQD}^{\top}] \nonumber\\
    &=(1-\rho^2)\text{vec}\left[\begin{pmatrix}
        1 & 0
    \end{pmatrix}\begin{pmatrix}
    \sigma^2_{\varepsilon} & \sigma_{\varepsilon,\eta}\\
    \sigma_{\varepsilon,\eta} & \sigma^2_{\eta}
\end{pmatrix}\begin{pmatrix}
        1\\
        0
    \end{pmatrix}\right]+\text{vec}\left[\begin{pmatrix}
        0 & 1
    \end{pmatrix}\begin{pmatrix}
    \sigma^2_{\varepsilon} & \sigma_{\varepsilon,\eta}\\
    \sigma_{\varepsilon,\eta} & \sigma^2_{\eta}
\end{pmatrix}\begin{pmatrix}
        0\\
        1
    \end{pmatrix}\right] \nonumber\\
    &=(1-\rho^2)\text{vec}[\sigma^2_{\varepsilon}]+\text{vec}[\sigma^2_{\eta}] \nonumber\\
    &=(1-\rho^2)\sigma^2_{\varepsilon}+\sigma^2_{\eta}. \label{Gamma0yAR1} 
\end{align}
Now, $\text{vec}[\bm{\Gamma}^1_y]$ and $\text{vec}[\bm{\Gamma}^2_y]$ are respectively given by
\begin{align}
    \text{vec}[\bm{\Gamma}^1_y] &= (\bm{C}\otimes\bm{CA})\text{vec}[\bm{\Gamma^0_x}] \nonumber\\
    &=\rho(1-\rho^2)\sigma^2_{\varepsilon}.
\end{align}
and
\begin{align}
    \text{vec}[\bm{\Gamma}^2_y] &= (\bm{C}\otimes\bm{CA}^2)\text{vec}[\bm{\Gamma^0_x}] \nonumber\\
    &=\rho^2(1-\rho^2)\sigma^2_{\varepsilon}.
\end{align}
The exhaustive summary is the infinite concatenation of the $\text{vec}[\bm{\Gamma}^h_y]$'s, that is 
\begin{equation}
    \bm{\kappa}=\begin{pmatrix}
        (1-\rho^2)\sigma^2_{\varepsilon}+\sigma^2_{\eta}\\
        \rho(1-\rho^2)\sigma^2_{\varepsilon}\\
        \rho^2(1-\rho^2)\sigma^2_{\varepsilon}\\
        \vdots 
    \end{pmatrix}.
\end{equation}
We restrict the exhaustive summary to the three (the number of parameters) first components, that is
\begin{equation}
    \bm{\kappa}_3=\begin{pmatrix}
        (1-\rho^2)\sigma^2_{\varepsilon}+\sigma^2_{\eta}\\
        \rho(1-\rho^2)\sigma^2_{\varepsilon}\\
        \rho^2(1-\rho^2)\sigma^2_{\varepsilon}
    \end{pmatrix}.
\end{equation}
We compute the derivative matrix 
\begin{equation}
    \dfrac{\partial \bm{\kappa}_3}{\partial \bm{\theta}}=\begin{pmatrix}
        -2\rho\sigma^2_{\varepsilon} & (1-3\rho^2)\sigma^2_{\varepsilon} & (2\rho-4\rho^3)\sigma^2_{\varepsilon}\\
        1 & 0 & 0\\
        1-\rho^2 & \rho(1-\rho^2) & \rho^2(1-\rho^2)
    \end{pmatrix}
\end{equation}
and we found the rank $r=3$ four three parameters. The derivative matrix $\partial \bm{\kappa}_3/\partial \bm{\theta}$ being full rank, the derivative matrix $\partial \bm{\kappa}/\partial \bm{\theta}=[
\partial \bm{\kappa}_3/\partial \bm{\theta},\dots ]$ is also full rank.

\subsubsection{Example with the AR(1) model with covariance parameters}

We consider the  AR(1) model where
\begin{equation}
    \bm{Q}=\begin{pmatrix}
    \sigma^2_{\varepsilon} & \sigma_{\varepsilon,\eta}\\
    \sigma_{\varepsilon,\eta} & \sigma^2_{\eta}
\end{pmatrix}.
\end{equation}
Now, we compute $\bm{\Gamma}^1_y$,
\begin{align}
    \text{vec}[\bm{\Gamma}^1_y] &= (\bm{C}\otimes\bm{CA})\text{vec}[\bm{\Gamma}^0_x]+\text{vec}\left[\bm{CBQD}^{\top}\right] \nonumber\\
    &=\rho(1-\rho^2)\sigma^2_{\varepsilon}+\text{vec}\left[\begin{pmatrix}
        1 & 0
    \end{pmatrix}\begin{pmatrix}
    \sigma^2_{\varepsilon} & \sigma_{\varepsilon,\eta}\\
    \sigma_{\varepsilon,\eta} & \sigma^2_{\eta}
\end{pmatrix}\begin{pmatrix}
        0\\
        1
    \end{pmatrix}\right] \nonumber\\
    &=\rho(1-\rho^2)\sigma^2_{\varepsilon}+\sigma_{\varepsilon,\eta}.
\end{align}
We recursively obtain
\begin{equation}
    \text{vec}[\Gamma^2_y]=\rho^2(1-\rho^2)\sigma^2_{\varepsilon}+\rho\sigma_{\varepsilon,\eta}
\end{equation}
and
\begin{equation}
    \text{vec}[\Gamma^3_y]=\rho^3(1-\rho^2)\sigma^2_{\varepsilon}+\rho^2\sigma_{\varepsilon,\eta}.
\end{equation}
Here we illustrate the necessity to use lag terms to take into account covariance parameters. A quick inspection of Equation \ref{Gamma0yAR1} would suggest that all elements of $\bm{Q}$ are always included in the variance-covariance matrix. However, because in most cases matrices $\bm{B}$ and $\bm{D}$ will be sparse, and $\bm{Q}$ occurs only as $\bm{BQB}^{\top}$ and $\bm{DQD}^{\top}$ in  Equation \ref{Gamma0y} and Equation \ref{Gamma0yAR1}, non-diagonal blocks $\bm{Q}_{12}$ and $\bm{Q}_{21}$ disappear from $\bm{\Gamma}_y^0$. Indeed, in this example we have $\bm{BQB}^{\top}=\sigma^2_{\varepsilon}$, $\bm{DQD}^{\top}=\sigma^2_{\eta}$ and $\bm{BQD}^{\top}=\sigma_{\varepsilon, \eta}$, but the latter product making the covariance between observation and process only appears in the $h>0$ lag terms. 

The exhaustive summary is the infinite concatenation of the $\text{vec}[\bm{\Gamma}^h_y]$'s, that is 
\begin{equation}
    \bm{\kappa}=\begin{pmatrix}
        (1-\rho^2)\sigma^2_{\varepsilon}+\sigma^2_{\eta}\\
        \rho(1-\rho^2)\sigma^2_{\varepsilon}+\sigma_{\varepsilon,\eta}\\
        \rho^2(1-\rho^2)\sigma^2_{\varepsilon}+\rho\sigma_{\varepsilon,\eta}\\
        \rho^3(1-\rho^2)\sigma^2_{\varepsilon}+\rho^2\sigma_{\varepsilon,\eta}\\
        \vdots 
    \end{pmatrix}
\end{equation}
For now, we cut the exhaustive summary in order to have the same number of rows as parameters
\begin{equation}
    \bm{\kappa}_4=\begin{pmatrix}
        (1-\rho^2)\sigma^2_{\varepsilon}+\sigma^2_{\eta}\\
        \rho(1-\rho^2)\sigma^2_{\varepsilon}+\sigma_{\varepsilon,\eta}\\
        \rho^2(1-\rho^2)\sigma^2_{\varepsilon}+\rho\sigma_{\varepsilon,\eta}\\
        \rho^3(1-\rho^2)\sigma^2_{\varepsilon}+\rho^2\sigma_{\varepsilon,\eta}
    \end{pmatrix}
\end{equation}
We compute the derivative matrix 
\begin{equation}
    \dfrac{\partial \bm{\kappa}_4}{\partial \bm{\theta}}=\begin{pmatrix}
        -2\rho\sigma^2_{\varepsilon} & (1-3\rho^2)\sigma^2_{\varepsilon} & (2\rho-4\rho^3)\sigma^2_{\varepsilon}+\sigma_{\varepsilon,\eta} & (3\rho^2-5\rho^4)\sigma^2_{\varepsilon}+2\rho\sigma_{\varepsilon,\eta}\\
        1 & 0 & 0 & 0\\
        1-\rho^2 & \rho(1-\rho^2) & \rho^2(1-\rho^2) & \rho^3(1-\rho^2)\\
        0 & 1 & \rho & \rho^2
    \end{pmatrix}
\end{equation}
and we calculate the rank $r=3$ for four parameters.  Solving the equation $\bm{\alpha}^{\top} (\partial \bm{\kappa}_4/\partial \bm{\theta})=\bm{0}$ gives 
\begin{equation}
    \bm{\alpha}^{\top}=\begin{pmatrix}
    0 & \dfrac{1}{\rho} & \dfrac{1}{(\rho^2-1)\rho} &  1
\end{pmatrix}.
\end{equation}
Solving the partial differential equation 
\begin{equation}
    \dfrac{1}{\rho}\dfrac{\partial f}{\partial \sigma^2_{\eta}}+\dfrac{1}{(\rho^2-1)\rho}\dfrac{\partial f}{\partial \sigma^2_{\varepsilon}}+\dfrac{\partial f}{\partial \sigma_{\varepsilon,\eta}}=0
\end{equation}
gives $\beta_1=\rho$, $\beta_2=\dfrac{\rho^2\sigma^2_{\varepsilon}-\sigma^2_{\eta}-\sigma^2_{\varepsilon}}{\rho^2-1}$ and $\beta_3=-\rho\sigma^2_{\eta}+\sigma_{\varepsilon,\eta}$ as estimable parameters. 

From the previous result, we want to deduce the rank of $\partial \bm{\kappa}/\partial \bm{\theta}$. We recall that we  have for $h>0$
\begin{equation}
    \text{vec}[\bm{\Gamma}^h_y]=\rho^h(1-\rho^2)\sigma^2_{\varepsilon}+\rho^{h-1}\sigma_{\varepsilon,\eta}
\end{equation}
Since, for $h>0$, we still have 
\begin{align*}
    &  \dfrac{1}{\rho}\dfrac{\partial \,\text{vec}[\bm{\Gamma}^h_y]}{\partial \sigma^2_{\eta}}+\dfrac{1}{(\rho^2-1)\rho}\dfrac{\partial \, \text{vec}[\bm{\Gamma}^h_y]}{\partial \sigma^2_{\varepsilon}}+\dfrac{\partial \, \text{vec}[\bm{\Gamma}^h_y]}{\partial \sigma_{\varepsilon,\eta}}\\
    =& \dfrac{1}{\rho}\times 0 + \dfrac{1}{(\rho^2-1)\rho}\times \rho^h(1-\rho^2)+\rho^{h-1}\\
    =&-\rho^{h-1}+\rho^{h-1}\\
    =&0
\end{align*}
there is a linear dependency between the lines of the matrix 
\begin{equation}
   \dfrac{\partial \bm{\kappa}}{\partial \bm{\theta}} =\begin{pmatrix}
        -2\rho\sigma^2_{\varepsilon} & (1-3\rho^2)\sigma^2_{\varepsilon} & \dots & \partial\, \text{vec}[\bm{\Gamma}^h_y]/\partial \rho & \dots\\
        1 & 0 & \dots & 0 & \dots\\
        1-\rho^2 & \rho(1-\rho^2) & \dots & \rho^h(1-\rho^2) & \dots\\
        0 & 1 & \dots & \rho^{h-1} & \dots
    \end{pmatrix}
\end{equation}
which has then rank $r=3$. As the derivative matrix has rank $r=3$, and there are four parameters, this model is parameter redundant with deficiency $d=1$. We recover here naturally the same result as the Fourier approach---although with a little more technicality. 

\subsubsection{Comparison with a previous approach}
In this section, we briefly recall the exhaustive summary proposed by \citet{cole2016parameter}. We consider a model of the form
\begin{equation}
    \bm{x}_{t+1}=\bm{Ax}_t+\bm{\varepsilon}_t \quad \text{and} \quad \bm{y}_t=\bm{Cx}_t+\eta_t,
\end{equation}
where $(\bm{\varepsilon}_t)$ and $(\bm{\eta}_t)$ are error processes with zero means. The exhaustive summary is given by 
\begin{equation}
    \bm{\kappa}=\begin{pmatrix}
        \mathbb{E}[\bm{y}_1]\\
        \mathbb{V}[\bm{y}_1]\\
        \mathbb{E}[\bm{y}_2]\\
        \mathbb{V}[\bm{y}_2]\\
        \vdots
    \end{pmatrix},
\end{equation}
where $\mathbb{E}[\bm{y}_t]=\bm{CA}^t\bm{x}_0$, with $\bm{x}_0\neq0$. If we go back to the AR(1) example, we obtain
\begin{equation}
    \bm{\kappa}=\begin{pmatrix}
        \rho x_0\\
    \sigma^2_{\varepsilon}+\sigma^2_{\eta}\\
    \rho^2 x_0\\
    (1+\rho^2)\sigma^2_{\varepsilon}+\sigma^2_{\eta}\\
    \rho^3 x_0\\
    (1+\rho^2+\rho^4)\sigma^2_{\varepsilon}+\sigma^2_{\eta}\\
    \vdots\\
    \rho^t x_0\\
    (1+\rho^2+\dots+\rho^{2(t-1)})\sigma^2_{\varepsilon}+\sigma^2_{\eta}\\
    \vdots
    \end{pmatrix}
\end{equation}
and we see that the exhaustive summary does not involve covariance parameters. This affirms the need to use an approach based on spectral density or autocovariance with lags to take into account covariance information in the identifiability diagnostic. Had we considered the model without covariance parameter, the same use of the Extension theorem as Section \ref{AR1Autocovariance} can be adopted, and we would reach the conclusion that the model is full rank using the variance-augmented summaries. Note well that it is not proven that the successive expectations and instantaneous variance-covariance matrices provide sufficient exhaustive summaries for state-space models without covariances between state and observation errors; it works for the state-space AR(1) model but it is not clear that it will work more generally. In particular it is important to note the different meaning of the term variance-covariance here and in \citet{cole2016parameter,aldrin2021caveats}: we consider a lag-zero autocovariance $\bm{\Gamma}_y^0$ computed at the stationary state, while they consider $\mathbb{V}[\bm{y}_t]$ that is computed at successive times (and does not require a stationary assumption). 

\subsection{Exhaustive summary based on the Kalman filter}
\label{KalmanFilterSection}

The stochastic linear discrete-time SSM we consider in this section consists in the two processes
\begin{equation}
    \bm{x}_{t+1}=\bm{Ax}_{t}+\bm{Bu}_{t}+\bm{\varepsilon}_{t},\quad \bm{\varepsilon}_{t}\sim \mathcal{N}(\bm{0}, \bm{Q})
\end{equation}
and 
\begin{equation}
    \bm{y}_t=\bm{Cx}_{t}+\bm{Du}_{t}+\bm{\eta}_{t},\quad \bm{\eta}_{t}\sim \mathcal{N}(\bm{0}, \bm{R})
\end{equation}
for $t=1,\dots,T$. In this model, the input sequence $(\bm{u}_t)$ is deterministic and could be considered as covariates. The parameters of the model $\bm{\theta}=\{\bm{A}, \bm{B}, \bm{C}, \bm{D}, \bm{Q}, \bm{R}\}$ can be estimated by the Kalman filter algorithm under the assumption that $\bm{\varepsilon}_{t}$ and $\bm{\eta}_{t}$ are Gaussian noise. The associated algorithm is recalled in the next section. 

\subsubsection{The Kalman Filter algorithm}

Following the derivations given by \citet{DeJong1988likelihood} and \citet{vlcek2001maximum}, the prediction and filtration steps of the Kalman filter are
\begin{align}
    \bm{z}_t &=\bm{y}_t-\bm{C}\bm{\hat{x}}_{t|t-1}-\bm{Du}_{t} \nonumber\\
    \bm{P}_t &=\bm{C}\bm{\Sigma}_{t|t-1}\bm{C}^{\top}+\bm{R} \nonumber\\
    \bm{K}_t &=\bm{\Sigma}_{t|t-1}\bm{C}^{\top}\bm{P}_t^{-1} \nonumber\\
    \bm{\hat{x}}_{t|t} &= \bm{\hat{x}}_{t|t-1}+\bm{K}_t \bm{z}_t \nonumber\\
    \bm{\Sigma}_{t|t} &= (\bm{I}-\bm{K}_t\bm{C})\bm{\Sigma}_{t|t-1} \nonumber\\
    \bm{\hat{x}}_{t+1|t} &= \bm{A}\bm{\hat{x}}_{t|t}+\bm{Bu}_t \nonumber\\
    \bm{\Sigma}_{t+1|t} &= \bm{A}\bm{\Sigma}_{t|t}\bm{A}^{\top}+\bm{Q}
\end{align}
where $\bm{z}_t$ is the innovation, $\bm{P}_t$ the innovation covariance, $\bm{K}_t$ the Kalman gain, $\bm{\hat{x}}_{t|t}$ the updated state, $\bm{\Sigma}_{t|t}$ the updated covariance, $\bm{\hat{x}}_{t+1|t}$ the predicted state and $\bm{\Sigma}_{t+1|t}$ the predicted covariance. $\bm{\hat{x}}_{1|0}$ and $\bm{\Sigma}_{1|0}$ are arbitrary initial values. The previous steps lead to the marginal likelihood of the observations.

\subsubsection{The marginal likelihood}

By the chain rule, the likelihood can be factored as the product of the probability of each observation given previous observations, that is
\begin{equation}
    L(\bm{y}_1,\dots,\bm{y}_T)=\prod_{t=1}^{\top} L(\bm{y}_t|\bm{y}_1,\dots,\bm{y}_{t-1})
\end{equation}
where
\begin{equation}
    \bm{y}_t|\bm{y}_1,\dots,\bm{y}_{t-1}\sim \mathcal{N}(\bm{C}\bm{\hat{x}}_{t|t-1}+\bm{Du}_{t}, \bm{P}_t),\quad \forall t\geq 1.
\end{equation}
Finally, after removing constants, minus twice the log-likelihood becomes
\begin{equation}
\label{likelihoodKF}
    l(\bm{y}_1,\dots,\bm{y}_T)=\sum_{t=1}^{\top} l(\bm{y}_t|\bm{y}_1,\dots,\bm{y}_{t-1})=\sum_{t=1}^{\top} \ln(\det(\bm{P}_t)) + \bm{z}_t^{\top}\bm{P}_t^{-1}\bm{z}_t.
\end{equation}
Thanks to the log-likelihood, we can build a suitable exhaustive summary.

\subsubsection{The Kalman Filter exhaustive summary}
First, we need to recall the extension theorem \citep[][Theorem 3]{cole2010determining}. A model represented by $\bm{\kappa}_1(\bm{\theta}_1)$ is extended by adding extra parameters $\bm{\theta}_2$ and new exhaustive summary terms $\bm{\kappa}_2(\bm{\theta}')$, with $\bm{\theta}'= (\bm{\theta}_1,\bm{\theta}_2)$. The extended model’s exhaustive summary is then $\bm{\kappa}(\bm{\theta}')^{\top} = (\bm{\kappa}_1(\bm{\theta}_1)^{\top},\bm{\kappa}_2(\bm{\theta}')^{\top})^{\top}$.  The derivative matrix of the extended model is
\begin{equation}
    \dfrac{\partial \bm{\kappa}}{\partial \bm{\theta'}}=\begin{pmatrix}
    \dfrac{\partial \bm{\kappa}_1}{\partial \bm{\theta}_1} & \dfrac{\partial \bm{\kappa}_2}{\partial \bm{\theta}_1}\\
    \bm{0} & \dfrac{\partial \bm{\kappa}_2}{\partial \bm{\theta}_2}
\end{pmatrix}.
\end{equation}
With these notations, \citet[][Theorem 3]{cole2010determining} show that
\begin{theorem}
    If the original model is full rank (i.e. $\partial \bm{\kappa}_1 / \partial \bm{\theta}_1$ is full rank) and $\partial \bm{\kappa}_2 / \partial \bm{\theta}_2$ is full rank, then the extended model is full rank also.
\end{theorem}

Now, as mentioned by \cite{cole2016parameter}, if a log-likelihood is written as $l(\bm{\theta})=\sum_{n=1}^N l_n(\bm{\theta})$, an exhaustive summary is 
\begin{equation}
\label{LikelihoodES}
    \bm{\kappa}_N(\bm{\theta})=\begin{pmatrix}
        l_1(\bm{\theta})\\
        \vdots\\
        l_N(\bm{\theta})
    \end{pmatrix}.
\end{equation}
Since in practice the number of observations $N$ can be (very) large, we want to reduce the length of $\bm{\kappa}_N$. Let $p$ denotes the number of parameters. Assuming that $N\geq p$, we have
\begin{theorem}
\label{TheoremLikelihood}
    For a model with the log-likelihood of the form $l(\bm{\theta})=\sum_{n=1}^N l_n(\bm{\theta})$, an exhaustive summary is given by
    \begin{equation}
        \bm{\kappa}_p(\bm{\theta})=\begin{pmatrix}
        l_1(\bm{\theta})\\
        \vdots\\
        l_p(\bm{\theta})
        \end{pmatrix}.
    \end{equation}
    Moreover, the model represented by $\bm{\kappa}_N$ is full rank if the model represented by $\bm{\kappa}_{p}$ is full rank.
\end{theorem}
\begin{proof}
    From Equation (\ref{LikelihoodES}), we know that $\bm{\kappa}_N$ is an exhaustive summary with $N$ elements. We can rewrite $\bm{\kappa}_N$ as $\bm{\kappa}_N^{\top}=(\bm{\kappa}_p^{\top},\bm{\kappa}_{p+}^{\top})^{\top}$ where $\bm{\kappa}_p = (l_1,\dots,l_p)^{\top}$ and $\bm{\kappa}_{p+} = (l_{p+1},\dots,l_N)^{\top}$. 
    
    First, we note that $\bm{\kappa}_p$ and $\bm{\kappa}_{p+}$ are well defined as  exhaustive summaries since they are the exhaustive summaries built by Equation (\ref{LikelihoodES}) for the adequate observations.

    Now, let $\partial \bm{\kappa}_p/\partial \bm{\theta}$ and $\partial \bm{\kappa}_{p+}/\partial \bm{\theta}$ be the derivative matrices of respectively $\bm{\kappa}_p$ and $\bm{\kappa}_{p+}$. The derivative matrix of $\bm{\kappa}_N$ writes
    \begin{equation}
        \dfrac{\partial \bm{\kappa}_{N}}{\partial \bm{\theta}}=
        \left(\begin{matrix}
        \dfrac{\partial \bm{\kappa}_{p}}{\partial \bm{\theta}} & \dfrac{\partial \bm{\kappa}_{p+}}{\partial \bm{\theta}}     
        \end{matrix}\right)\in\mathbb{R}^{p\times N}.
    \end{equation}
    As an application of the extension theorem, \citet[][page 53]{cole2020parameter} remarks that if $\partial \bm{\kappa}_p/\partial \bm{\theta}$ is full rank, then $\partial \bm{\kappa}_N/\partial \bm{\theta}$ is full rank. This concludes the proof.
\end{proof}
In this work, we aim to build an exhaustive summary based on the Kalman filter. Assuming that the number of time points $T$ is greater than the number of parameters $p$, we have
\begin{corollary}
\label{KalmanFilterExhaustiveSummary}
    For stochastic linear-time-invariant discrete-time state-space models, the Kalman Filter exhaustive summary is given by
    \begin{equation}
        \bm{\kappa}_{KF}(\bm{\theta})=\begin{pmatrix}
        l_1(\bm{\theta})\\
        \vdots\\
        l_p(\bm{\theta})
        \end{pmatrix},
    \end{equation}
    where $l_t(\bm{\theta})=\ln(\det(\bm{P}_t))+\bm{z}_t^{\top}\bm{P}_t^{-1}\bm{z}_t$, for all $t=1,\dots,p$.
\end{corollary}
\begin{proof}
    This result is a direct application of Theorem \ref{TheoremLikelihood} when the log-likelihood derives from the Kalman filter as given in Equation (\ref{likelihoodKF}).
\end{proof}

\subsubsection{Application to a frequently used AR(1) model}

In this section, we want to apply Corollary \ref{KalmanFilterExhaustiveSummary} on the AR(1) model with intercept presented by \citet{knape2008estimability} (the same model is presented without intercept in \citealt{AugerMethe2016state}). The model is defined by the two following processes 
\begin{equation}
    x_{t+1}=a+cx_t+\varepsilon_t, \quad \varepsilon_t\sim \mathcal{N}(0,\sigma^2)
\end{equation}
and
\begin{equation}
    y_t=x_t+\eta_t,\quad \eta_t\sim \mathcal{N}(0,\tau^2),
\end{equation}
where $a$ is the intercept, $c$ the autoregressive parameter and $y_t$ the log-observed population abundance. Here, $(u_t)$ is a constant sequence equal to one. The vector of parameters is $\bm{\theta}=(a,c,\sigma^2,\tau^2)$. Taking $\hat{x}_{1|0}=0$ and $\Sigma_{1|0}=1$ as initial values, the exhaustive summary is given by the four first evaluations of the log-likelihood
\begin{align}
    \bm{\kappa}_{KF}(\bm{\theta})&=\begin{pmatrix}
    l(y_1)\\
    l(y_2\mid y_1)\\
    l(y_3\mid y_1,y_2)\\
    l(y_4\mid y_1,y_2,y_3)
\end{pmatrix} \nonumber\\
    &=\begin{pmatrix}
    \ln(1+\tau^2)+\dfrac{y_1^2}{1+\tau^2}\\
    \ln\left(c^2\left(1-\dfrac{1}{1+\tau^2}\right)+\sigma^2+\tau^2\right)+\dfrac{\left(y_2-\dfrac{cy_1}{1+\tau^2}-a\right)^2}{c^2\left(1-\dfrac{1}{1+\tau^2}\right)+\sigma^2+\tau^2}\\
    l(y_3\mid y_1,y_2)\\
    l(y_4\mid y_1,y_2,y_3)
\end{pmatrix},
\end{align}
where $l(y_3\mid y_1,y_2)$ and $l(y_4\mid y_1,y_2,y_3)$ are too long to be written here. The derivative matrix has rank $r=4$, and there are four parameters, this model is not parameter redundant.

\end{document}